\documentclass[runningheads]{llncs}
\usepackage{algorithm,algpseudocode}
\usepackage{mathtools}

\usepackage{graphicx}
\usepackage{tabularx}
\usepackage{textcomp}
\usepackage{xcolor}
\usepackage{enumitem}
\usepackage{float}
\usepackage{wrapfig}
\usepackage{booktabs}

\usepackage{pgf,tikz,pgfplots}
\pgfplotsset{compat=1.15}
\usepackage{mathrsfs}
\usepackage{hyperref}
\usepackage{multicol}
\usepackage{multirow}
\usepackage{amssymb}
\usepackage{xcolor}
\usepackage[linesnumbered,ruled,vlined,algo2e]{algorithm2e}
\usepackage[export]{adjustbox}

\SetCommentSty{mycommfont}

\newenvironment{proofsketch}{%
  \proof}{\endproof}

\usepackage{orcidlink}

\begin{document}
\title{
Time-optimal Asynchronous Minimal Vertex Covering by Myopic Robots on Graph
}

\author{Saswata Jana\orcidlink{0000-0003-3238-8233} \thanks{Supported by  Prime Minister's Research Fellowship (PMRF) scheme of the Govt. of India (PMRF-ID: 1902165)} \and Subhajit Pramanick\orcidlink{0000-0002-4799-7720} \and Adri Bhattacharya\orcidlink{0000-0003-1517-8779} \thanks{Supported by CSIR, Govt. of India, Grant Number: 09/731(0178)/2020-EMR-I} \and Partha Sarathi Mandal\orcidlink{0000-0002-8632-5767}}
\authorrunning{Jana et al.}
\titlerunning{Time-optimal Asynchronous Minimal Vertex Covering by Myopic Robots}
\institute{Indian Institute of Technology Guwahati, Assam, 781039, India}
\email{}

\maketitle

\begin{abstract}
In a connected graph with an autonomous robot swarm with limited visibility, it is natural to ask whether the robots can be deployed to certain vertices satisfying a given property using only local knowledge.
This paper affirmatively answers the question with a set of \emph{myopic} (finite visibility range) luminous robots with the aim of \emph{filling a minimal vertex cover} (MVC) of a given graph $G = (V, E)$.
The graph has special vertices, called \emph{doors}, through which robots enter sequentially. Starting from the doors, the goal of the robots is to settle on a set of vertices that forms a minimal vertex cover of $G$ under the asynchronous ($\mathcal{ASYNC}$) scheduler. We are also interested in achieving the \emph{minimum vertex cover} (MinVC, which is NP-hard \cite{Karp1972} for general graphs) for a specific graph class using the myopic robots. 
We establish lower bounds on the visibility range for the robots and on the time complexity (which is $\Omega(|E|)$).
We present two algorithms for trees: one for single door, which is both time and memory-optimal, and the other for multiple doors, which is memory-optimal and achieves time-optimality when the number of doors is a constant.
Interestingly, our technique achieves MinVC on trees with a single door. 
We then move to the general graph, where we present two algorithms, one for the single door and the other for the multiple doors with an extra memory of $O(\log \Delta)$ for the robots, where $\Delta$ is the maximum degree of $G$.
All our algorithms run in $O(|E|)$ epochs.  
\end{abstract}


\section{Introduction}

\noindent \textbf{Background and Motivation.}
Efficient sensor deployment is crucial in wireless sensor systems, especially in areas modeled as graphs, where coverage or placement must satisfy specific optimization goals.
While static sensors rely on manual or random placement, which may not meet the coverage goals, mobile sensors, in contrast, can self-deploy without any central coordination. However, such autonomous behaviour requires complex, localized algorithm design, a central challenge in this line of research \cite{1200625,10.1023/A:1019625207705,Hsiang2004,1307146}.  
In this area, low-resource swarm robots \cite{10.1007/3-540-46632-0_10} have been drawing the attention of researchers for many decades now. Under the classical \emph{Look-Compute-Move} framework (which is formally defined later), many seminal problems have been studied in graphs, such as filling graph vertices \cite{10.1007/978-3-540-92862-1_11,10.1007/978-3-642-45346-5_17,hideg2017uniform,hideg2020asynchronous}, maximal independent set and dominating set \cite{barriere2011uniform,10.1007/978-3-031-48882-5_10,ELOR2011783,10.1145/3631461.3631543},  etc. 
In this context, a fundamental question emerges: \emph{Can we achieve a specific graph property using autonomous mobile robots with local information?}
In this paper, we affirmatively answer this question by exploring the vertex cover problem. 

The \textit{vertex cover} of a graph is a set of vertices such that every edge is incident to (or covered by) at least one vertex from the set. 
Since \emph{minimum vertex cover} (MinVC) for general graph is NP-hard, we shift our attention to achieving a \emph{minimal vertex cover} (MVC) of an arbitrary connected graph using a set of \emph{autonomous} (no central control), \emph{myopic} (having limited visibility range), \emph{anonymous} (without identification) and \emph{homogeneous} (the algorithm is same for all) mobile robots working under $\mathcal{ASYNC}$ scheduler.
 MVC is a computational problem with many theoretical and real-life applications, such as an area with multiple lanes, modelled as a graph needs to be minimally covered by guards.

In this paper, we study a constrained variant of the self-deployment problem, known as \emph{filling MVC problem}, where robots, starting from specific \emph{door} vertices, need to collaboratively fill a subset of vertices that forms an MVC of the graph.
The \emph{filling problem} using mobile robots, in which the robots are injected one at a time (using doors) into an unknown environment, has been introduced by Hsiang et al. \cite{Hsiang2004}. 
A seminal result by Barrameda et al.
\cite{10.1007/978-3-540-92862-1_11} proved that the filling problem is impossible to solve deterministically using oblivious robots, even with unlimited visibility. 
To circumvent such impossibility, they considered persistent memory in the form of colors and proposed collision-free algorithms (two collocated robots make a \emph{collision} that could damage the robots or the equipment) for robots with constant visibility range. 
We consider the bare-bone model in solving the filling MVC problem in the optimal time, considering $\mathcal{ASYNC}$ setting (the scheduler with the least assumptions), limited visibility with visibility range not depending on any graph parameter, no knowledge of the graph, luminous robots and no collision. 
Each robot has a persistent memory in the form of an externally visible persistent light that can flash a color from a predetermined color set.  
We judge algorithms on four metrics: (i) memory requirement, (ii) number of colors, (iii) visibility range, and (iv) time required to achieve the goal.
Moreover, our interest lies in identifying some known graph classes where we can achieve the MinVC in optimal (or near-optimal) time.

\noindent \textbf{Related Works.}
The filling problem is introduced by Hsiang et al. \cite{Hsiang2004} in a setting where robots enter a region (modeled as a connected subset of an integer grid) through doors, aiming to occupy every cell. 
Barrameda et al.
\cite{10.1007/978-3-540-92862-1_11} showed that the problem is impossible with oblivious robots and thus considered persistent memory (in the form of colors) to solve the filling problem on a connected orthogonal space, partitioned into square cells (which they later model using graphs) with special door cells.
They present two algorithms under $\mathcal{ASYNC}$: one for the single door, where sensors have $1$ hop visibility and $2$ bits of persistent memory, and the other for multiple doors, where sensors have $2$ hops visibility and $O(1)$ persistent memory.
Barrameda et al.  \cite{10.1007/978-3-642-45346-5_17} later extended the problem on orthogonal regions with holes using robots having $6$ hops visibility. 
Under $\mathcal{ASYNC}$, Kamei and Tixeuil \cite{kamei2020asynchronous} solved the maximum independent set (Max\_IS) filling problem on a finite grid with a door at a corner using $3$ colors and $2$ hop visibility with port-labelling and however it needs $7$ colors and $3$ hop visibility without it.  
Hideg and Lukovszki  \cite{hideg2017uniform} studied the filling problem in orthogonal regions with doors. They propose two algorithms in $O(n)$ time under synchronous setting, one for single door and the other for multiple doors, where $n$ is the number of cells.
Later in \cite{hideg2020asynchronous}, they presented the filling problem for an arbitrary connected graph under the asynchronous setting. 
The deployment of mobile robots on graph vertices with specific properties later gained popularity.
Pramanick et al. \cite{PRAMANICK2023114187} presented two algorithms for the maximal independent set (MIS) filling problem on arbitrary graphs: one for single door under $\mathcal{ASYNC}$ using $3$ hop visibility and $O(\log \Delta)$ memory and another for multiple doors under semi-synchronous setting using $5$ hop visibility and $O(\log (\Delta +k))$ memory, where $\Delta$ is the maximum degree and $k$ is the number of doors.
Pattanayak et al. \cite{10.1145/3631461.3631543} proposed algorithms for finding an MIS using mobile agents communicating face-to-face with each other. 
Chand et al. \cite{10.1007/978-3-031-48882-5_10} investigated the problem of placing mobile agents on a dominating set of the given graph.

\begin{table}
\begin{center}
\resizebox{\columnwidth}{!}{%
{\renewcommand{\arraystretch}{2}
\begin{tabular}{|c|c|c|c|c|c|c|c|}
		\hline
           \textbf{No of} & \textbf{Graph} & \textbf{Activation} & \textbf{Visibility} & \textbf{No. of} & \textbf{Memory} & \textbf{Lower} & \textbf{Time}\\
           \textbf{Doors} & \textbf{Property} & \textbf{Scheduler} & \textbf{Range} & \textbf{Colors} & \textbf{(in bits)} & \textbf{Bound} & \textbf{(in epoch)}\\
		      \hline 
            \multirow{5}{2em}{\centering 1} & Graph Filling \cite{10.1007/978-3-540-92862-1_11} & $\mathcal{ASYNC}$ & 1 & - & $O(1)$ & - & - \\
            \cline{2-8}
            & Graph Filling \cite{hideg2020asynchronous} & $\mathcal{ASYNC}$ & 1 & $O(\Delta)$ & $O(\log \Delta)$ & - & $O(|V|^2)$\\ 
            \cline{2-8}
            & MAX\_IS (Grid) \cite{kamei2020asynchronous} & $\mathcal{ASYNC}$ & 2 & $O(1)$ & - & - & $O(|V|(L+l))$\\ 
            \cline{2-8}
            & Filling MIS \cite{PRAMANICK2023114187} & $\mathcal{ASYNC}$ & 3 & $O(\Delta)$ & $O(\log \Delta)$ & $ \Omega (|V|)$ & $O(|V|^2)$ \\ 
            \cline{2-8}
            & Filling MVC (Th. \ref{thm:asyncGenSS}) & $\mathcal{ASYNC}$ & 3 & $O(1)$ & $O(\log \Delta)$ & $\Omega(|E|)$ & $O(|E|)$\\  
            \cline{2-4}
           \hline
           \multirow{2}{4em}{\centering $H(>1)$} & Graph Filling \cite{10.1007/978-3-540-92862-1_11} & $\mathcal{ASYNC}$ & 2 & $O(H)$ & $O(1)$ & - & - \\ 
            \cline{2-8}
            & Graph Filling \cite{hideg2020asynchronous} & $\mathcal{ASYNC}$ & 2 & $O(\Delta + H)$ & $O(\log \Delta)$ & - & -\\ 
            \cline{2-8}
           & Filling MIS \cite{PRAMANICK2023114187} & $\mathcal{SSYNC}$ & 5 & $O(\Delta + H)$ & $O(\log (\Delta + H))$ & - & $O(|V|^2)$ \\ 
           \cline{2-8}
            & Filling MVC (Th. \ref{theorem:Graph_MultiDoor}) & $\mathcal{ASYNC}$ & 4 & $O(H)$ & $O(\log \Delta)$ & $\Omega(|E|- H)$ & $O(|E|)$\\ 
            \cline{2-4}
           \hline
\end{tabular} }
}
\end{center}
\caption{Comparison with the Literature. $\Delta$ is the maximum degree of the graph, and $m$ is the size of the maximal independent set (MIS). MAX\_IS stands for maximum independent set. $L$ and $l$ are the dimensions of the grid.  $\mathcal{SSYNC}$ stands for semi-synchronous scheduler.}
\label{tab:contribution}
\end{table}

\noindent \textbf{Contributions:}
Our contributions in this paper are the following:

\begin{itemize}[left=0pt]
    \item We show impossibility to solve the filling MVC problem with 1 hop visibility (Theorem \ref{thm:lower-bound-visibility-range}) and a lower bound of $\Omega(|E|)$ on time complexity (Theorem \ref{thm:time_lowerBound}).

    \item Algorithm: \textsc{Tree\_SingleDoor}, requires $2$ hops of visibility and $4$ colors. Interestingly, this algorithm achieves MinVC on trees in optimal time $O(|E|)$ and optimal $O(1)$ memory (Theorem \ref{theorem:tree-singledoor}).

    \item Algorithm: \textsc{Tree\_MultiDoor} ($H$ doors) requires $4$ hops of visibility and $O(H)$ colors (which is also optimal) and runs in $O(|E|)$ epochs (Theorem \ref{thm:coll_tree_MS}).

    \item Later, we present two algorithms for a general graph: a time-optimal algorithm: \textsc{Graph\_SingleDoor} with $3$ hops of visibility and $4$ colors, and another \textsc{Graph\_MultiDoor} algorithm (for $H$ doors) with $4$ hops of visibility and $O(H)$ colors. 
    Both algorithms on general graphs use additional $O(\log \Delta)$ memory and run in $O(|E|)$ epochs (Theorem \ref{thm:asyncGenSS} and \ref{theorem:Graph_MultiDoor}) without the knowledge of any graph parameter.
\end{itemize}


\section{Model and Preliminaries}
\label{sec:model}

\noindent \textbf{Graph Model:} We consider an anonymous graph $G(V, E)$ with $V$ as the set of vertices and $E$ as the set of edges, where the vertices do not possess any unique ID. 
However, the ports corresponding to the incident edges of a vertex $v$ are assigned a label from the set $\{ 1,2,\cdots, \delta(v)\}$, where $\delta (v)$ is the degree of $v$.
The edge between vertices $u$ and $v$ is denoted by $e(u,v)$ and $\Delta$ denotes the maximum degree of the graph.
A vertex has neither memory nor computational ability. 
An edge between two vertices receives independent port numbering at either end.
We use $p_v$ to represent a port incident to the vertex $v$.
When we say ``port $p$ of the edge $e(u, v)$'', we refer to the port of the edge $e(u,v)$ incident to $u$.
A vertex is called occupied if there is a robot positioned on it; otherwise unoccupied.

\noindent\textbf{Door:} 
The graph has some special vertices, referred to as \emph{doors}, where a group of robots are initially collocated and no two adjacent vertices are doors.
The robots on a door are serialized in a queue by a pre-processing algorithm (probably by any popular randomized leader election protocol, which is beyond our concern). 
When we say that the robot situated on a door executes the algorithm, we mean that the robot at the front of the queue executes the algorithm. When this robot leaves the door, the next robot in the queue comes to the front.
The color of the robot on a door indicates the color of the robot at the front.
A robot can identify whether it is at the door or not. 
We use $H$ to denote the number of doors.

\noindent\textbf{Robot Model:} 
The robots are autonomous, myopic, anonymous and homogeneous.
A robot at a vertex $v$ with visibility range $\phi$ hops can see all the vertices and their associated port numbers along each path of length $\phi$.
Additionally, the robots are \emph{luminous}, each having a light that determines a color from a predefined color set. 
A robot can see the color of all the robots in its visibility range, including itself. 
$r.color$ represents the color of the light for the robot $r$. At any given time, at most one robot can be located at a vertex, except for the door.

\noindent\textbf{Activation Cycle:}
Each robot operates in the classical \emph{Look-Compute-Move} (LCM) cycles.
\textit{Look:} The robot takes a snapshot of all visible vertices and the colors of the robots occupying them.
\textit{Compute:} It runs the algorithm using the snapshot to select a neighbouring vertex as its target or remains in place. It updates its color at the end of this phase if needed. 
\textit{Move:} It moves to the target vertex, if any.
A movement is not instantaneous.  

\noindent\textbf{Activation Scheduler:}
We consider the activation scheduler as asynchronous ($\mathcal{ASYNC}$), where the adversary is the strongest among all other schedulers.
Under $\mathcal{ASYNC}$, any robot can be activated at any time and may be idle for an arbitrarily large number of LCM cycles. However, the robot must be activated infinitely often (fairness condition). The time is measured in terms of \emph{epochs}, the smallest time interval in which each robot gets activated and executes a full LCM cycle at least once.

\noindent \textbf{Problem (\textit{Filling MVC}):} We are given an anonymous port-labelled connected graph $G(V, E)$ with $H$ door vertices, where luminous and myopic mobile robots operate. 
The objective is to settle the robots without collision on a set of vertices $V_1 \subset V$, such that $V_1$ forms a minimal vertex cover of $G$, all settled robots adopt a designated color to indicate termination so that no robot moves thereafter.

\section{Impossibility, Lower Bound and Technical Overview}
\label{sec:technical-difficulty}

We first highlight the lower bounds on visibility range and the time complexity.

\begin{theorem}
    \label{thm:lower-bound-visibility-range}
    Under $\mathcal{ASYNC}$, the filling MVC problem cannot be solved without collisions using 1-hop visibility, even with robots having unbounded memory.
\end{theorem}

\begin{proofsketch}
Consider a simple path of $5$ vertices, labeled $v_1$ through $v_5$, with a door at $v_1$. 
Vertex labels are provided for exposition only; the vertices themselves are anonymous.
To cover the edge $e(v_4, v_5)$, a robot $r_1$ must eventually reach $v_3$ (by taking the edge $e(v_1,v_2)$ and then $e(v_2,v_3)$).
If $r_1$ at $v_3$, moves back to $v_2$, the adversary can activate another robot $r_2$ in sync at the door, which moves to $v_2$, causing a collision.
To avoid revisiting $v_2$, if $r_1$ uses memory to store port numbers in the direction of $v_2$, it moves to $v_4$ to cover $e(v_4, v_5)$ from $v_3$.
After the movement, the adversary keeps $r_1$ idle at $v_4$ and activates $r_2$ at $v_2$.
In this case, $r_2$ moves to $v_3$ from $v_2$ as it cannot move to the door at $v_1$ (since there is a robot on it).
Next, with $r_1$ at $v_4$ and $r_2$ at $v_3$, the adversary again activates a third robot $r_3$ at the door.
$r_3$ moves to $v_2$ as it cannot see any other robot with its $1$ hop visibility from $v_1$.
If $r_1$ settles at $v_4$, the minimality is violated. 
If instead $r_1$ moves and settles at $v_5$, the adversary can similarly direct $r_2$ to $v_4$ and $r_3$ to $v_3$, again violating minimality under $\mathcal{ASYNC}$ and $1$ hop visibility. \qed
\end{proofsketch}

\begin{theorem}
    \label{thm:time_lowerBound}
    The time needed for an algorithm to solve the filling MVC problem on a graph $G$ with a single door by myopic mobile robots is $\Omega(|E|)$. 
\end{theorem}

\begin{proof}
    Let $G$ be a path with a door at one of the endpoints.
    The first robot that leaves the door takes at least $|E|- 1$ epochs to reach the edge incident to the other endpoint and cover it.
    \qed
\end{proof}

\begin{remark}
    \label{remark:impossbility-on-color-multidoor}
    An argument similar to that in \cite[Theorem 3]{10.1007/978-3-540-92862-1_11} shows that $\Omega(H)$ colors are necessary for the robots when there are $H > 1$ doors.
\end{remark}

\noindent\textbf{Challenges and Our Technique.} 
We begin by outlining the structure of our algorithms, which form the foundation for all subsequent sections. 
One can draw an analogy to depth-first search (DFS) graph traversal to have an intuitive understanding of our approach.
The robots exploring the graph maintain a virtual chain-like structure (a similar idea is also used in \cite{PRAMANICK2023114187}, but failed to achieve time optimality because of the highly sequential movement of the robots), primarily helping them avoid collision. 
Our approach achieves both efficiency and time-optimality by enabling every alternate robot, starting from the head of the chain to the door, to move within a constant number of epochs, thereby introducing a parallelism in movement. 
For any two successive robots in the chain, referred to as the predecessor and the successor, we ensure that under $\mathcal{ASYNC}$, each of them remains (at least) $2$ hops from its successor and (at most) $3$ hops from its predecessor within constant epochs.
This guarantees that such robots can move toward their predecessor in nearly constant epochs. 

A robot primarily follows its predecessor until the predecessor can no longer proceed. Throughout the process, we ensure that the chain of robots never crosses itself and that chains originating from different doors remain disjoint. Additionally, to prevent collisions, a robot $r$ that moves from a vertex $v$ to a neighbouring vertex $u$ is prohibited from returning to $v$.
In the multiple-door case, to circumvent the impossibility shown in \cite{10.1007/978-3-540-92862-1_11} (Theorem 3), we use distinct colors for the doors to establish a hierarchy among the robots entering through different doors, as two robots from two different doors having the same target vertex need to break the symmetry between them.
In the case of a general graph, we use additional memory to maintain the predecessor-successor relationship between two consecutive robots and store their locations, as there can be multiple paths between two vertices.
To give a structural overview of our algorithm, a brief execution example of the algorithm for the tree single-door case can be found in Section \ref{subsec-app:execution-example}.

\section{Algorithm for Filling MVC on Trees}
\label{sec:tree}

Prior work on related problems, such as maximal independent set \cite{PRAMANICK2023114187} and vertex filling \cite{hideg2020asynchronous}, requires $O(\log \Delta)$ persistent memory. Similarly, our solution in Section \ref{sec:general_graph} for MVC on general graph also uses $O(\log \Delta)$ memory.
This raises a natural question: \emph{Can we eliminate this memory requirement for specific graph classes, while preserving time optimality, at least for single-door}? Additionally, \emph{how close can we get to the MinVC on such graphs?}
We answer both questions affirmatively for trees. 
We begin with the following conventions.

\begin{definition}{(\textbf{$k$ hops neighbour of $v$})}
    A vertex with a shortest path of length $k$ from the vertex $v$ is called a $k$ hops neighbour of $v$. We denote it by $v^{k}_{nbr}$.
\end{definition}

\noindent When we say that the vertex $v'$ is a ``\emph{$k$ hop neighbour of $v$ along the port} $p_v$", we mean that there is a shortest path of length $(k-1)$ from $v_{nbr}^1$ to $v'$, where $v_{nbr}^1$ is the neighbour of $v$ along the port $p_v$.
We also define the following predicates.

\noindent $\blacktriangleright$~$\boldsymbol{occupied(v_{nbr}^k, p_v, \texttt{COL})}$:  This means that there exists a $k$ hops neighbour $v_{nbr}^k$ of the vertex $v$ along the port $p_v$, which is occupied by a robot with the color $\texttt{COL}$. When we write $\sim$ in place of \texttt{COL}, it indicates that the occupying robot can have any color from its color set. If we write $\neq\texttt{COL}$ instead of \texttt{COL}, it means that the occupying robot can have any color except \texttt{COL}.

\noindent $\blacktriangleright$~$\boldsymbol{\neg occupied(v_{nbr}^k, p_v, \sim)}$:  The vertex $v_{nbr }^k$ exists along $p_v$, but unoccupied.

\noindent $\blacktriangleright$~$\boldsymbol{exists(p_v, k)}$: There exists a $k$ hops neighbour of $v$ along the port $p_v$.
    
\noindent $\blacktriangleright$~$\boldsymbol{\neg exists(p_v, k)}$: There does not exist a $k$ hops neighbour of $v$ along $p_v$.

 For example, $\exists p_v \exists v_{nbr}^2 ~ occupied(v_{nbr}^2, p_v, \neq \texttt{FINISH}))$ means that there is a port $p_v$ (incident to $v$) along which there exists a 2 hops neighbour $v_{nbr}^2$ of $v$ such that the vertex $v_{nbr}^2$ is occupied by a robot, say $r'$ with $r'.color \neq $ \texttt{FINISH}.

\subsection{Algorithm (\textsc{\textmd{Tree\_SingleDoor}}) for Single Door}
\label{subsec:tree-singledoor}

In this algorithm, robots use 2 hops of visibility and four colors: \texttt{OFF-0}, \texttt{OFF-1}, \texttt{OFF-2} and \texttt{FINISH}. Initially, all the robots have the color \texttt{OFF-0}. A robot identifies its follower within its visibility range before executing a movement. 
A robot $r$ with the color \texttt{OFF-i} (\texttt{i} $\in \{0,1,2 \}$) considers the visible robot $r'$ with $r'.color = $\texttt{OFF-j} as its follower, where \texttt{j =i+1 (mod 3)}. 
In each LCM cycle, this decision is made solely based on colors, without storing the follower’s position.
We now define an eligible port for a robot $r$ situated at a vertex $v$.
\begin{definition}{\textbf{(Eligible Port from $v$)}}
A port $p_v$ incident with $v$ is called an eligible port for $r$ positioned at $v$ if the following are satisfied.
(i) An unoccupied $v_{nbr}^2$ exists along the port $p_v$.
(ii) There is no robot $r'$ with $r'.color \neq $ \texttt{FINISH} present at $2$ hops away from $v$ along the port $p_v$. 
    
\end{definition} 

Based on the position of $r$ and its current color, we differentiate the following two cases. In all the cases, $r$ does nothing when it sees a robot on the edges.

\noindent $\blacktriangleright$~ \textbf{Case 1 ($r$ is situated on the door):} 
    We further identify two sub-cases.
    \begin{itemize}[left=0pt]
        \item \textbf{Case 1.1} ($r.color =$ \texttt{OFF-0}): If $\exists p_v ~ occupied(v_{nbr}^1, p_v, \texttt{OFF-i}) \lor \exists p_v \exists v_{nbr}^2~ $ $occupied(v_{nbr}^2, p_v, \texttt{OFF-i})$ for $\texttt{i}\in \{0,1,2\}$, $r$ updates its color to \texttt{OFF-j}, where \texttt{j = i+1 (mod3)}.

        \item \textbf{Case 1.2} ($r.color =$ \texttt{OFF-1} or \texttt{OFF-2}): Since $r$ is on the door and with the color \texttt{OFF-1} or \texttt{OFF-2}, it must encounter a robot $r'$ at $1$ or $2$ hops away from it in an earlier LCM cycle. Such a robot $r$ is eligible to move either when $r'$ changes its color to \texttt{FINISH} or moves further. Thus, $r$ does nothing, if $\exists p_v~ (occupied(v_{nbr}^1, p_v, \neq \texttt{FINISH}) \lor \exists v_{nbr}^2 ~ occupied(v_{nbr}^2, p_v, \neq \texttt{FINISH}))$. 
    \end{itemize}
     
     \noindent 
     If neither of the predicates holds, then regardless of $r$ having the color \texttt{OFF-0}, \texttt{OFF-1} or \texttt{OFF-2}, it identifies the minimum eligible port (has the minimum label) $p_v^{min}$ from the current vertex $v$. If $p_v^{min}$ exists, $r$ moves to $v_{nbr}^1$, the neighbour of $v$ along the port $p_v^{min}$, with the current color. 
     While $r$ switches to \texttt{FINISH}, if $p_v^{min}$ does not exist, and $\exists p_v \neg occupied(v^1_{nbr}, p_v, \sim)$.

    \noindent $\blacktriangleright$~ \textbf{Case 2 ($r$ is not on the door \& $r.color =$ \texttt{OFF-i}, $\texttt{i} \in \{ 0,1,2\}$):}
    $r$ remains stationary till it does not see a robot $r'$ with color \texttt{OFF-j} (the follower of $r$), where \texttt{j = i+1 (mod 3)}.
    It also remains in place till it finds a robot $r'$ with color \texttt{OFF-j'}, where \texttt{j' = i-1 (mod 3)}.
    Otherwise, $r$ checks whether there is any eligible port incident to the current vertex $v$. 
    If no such port exists, it changes its color to \texttt{FINISH}. If it exists, it selects $v_{nbr}^1$ along the minimum eligible port $p_v^{min}$ as its target. Finally, $r$ moves to $v_{nbr}^1$, maintaining the current color.

    \noindent $\blacktriangleright$~\textbf{Termination Condition:} The algorithm terminates either when the door has a \texttt{FINISH}-colored robot or all the neighbours of the door are occupied with \texttt{FINISH}-colored robots.

\noindent \textbf{Analysis of the algorithm \textsc{Tree\_SingleDoor}:}
We now analyse the correctness and time-complexity of the algorithm with detailed proofs. For our convenience, we define the following.
We call a robot $r$ \emph{active}, if $r.color \neq $ \texttt{FINISH}. 
An active robot is a \emph{head} if it entered the graph before all other active robots in the current configuration.
We use a self-explanatory notation $distance(v_i, v_j)$ to denote the length of the path between the two vertices $v_i$ and $v_j$. Similar notation will also be used to denote the distance between the positions of two robots.

\vspace{1.5mm}

\begin{lemma}
    \label{lem:lie_same_subtree}
    If $r_1, r_2, \cdots, r_k$ are the active robots, each of which has crossed a vertex $v$ successively, they must lie on the same subtree rooted at $v$.
\end{lemma}

\begin{proof}
    \noindent 
    We first prove the case for the two robots $r_1$ and $r_2$. 
    Since $r_1$ and $r_2$ are successive robots, they are of color \texttt{OFF-i} and \texttt{OFF-j}, respectively, for some $i \in \{0,1,2\}$, where $j=(i+1)(mod~3)$. Each robot moves to its neighbour after computing the eligible port from its current vertex. Let $r_1$ leave the vertex $v$ through the port $p_v$, which is the minimum eligible port at the time of its movement. $r_1$ first moves to $v^1_{nbr}$ along $p_v$ and then to $v_{nbr}^2$. When $r_2$ reaches $v$ and finds $r_1$ at $v_{nbr}^2$, it remains in place until $r_1$ moves further to $v_{nbr}^3$. Once $r_1$ moves to $v_{nbr}^3$, $r_2$ lying at $v$ calculates the minimum eligible port. Since no other robots have left the vertex $v$ after $r_1$, all the ports that were eligible during $r_1$'s LCM cycle, where it moved to $v_{nbr}^1$ from $v$, remain eligible at the current time. Therefore, $p_v$ is the minimum eligible port for $r_2$ from $v$ and so it proceeds to $v_{nbr}^1$. Hence, both $r_1$ and $r_2$ lie on the same subtree rooted at $v$.

    Using a similar argument, we can prove that $r_2$ and $r_3$ lie on the same subtree rooted at $v$. In general $r_l$ and $r_{l+1}$ lie on the same subtree rooted at $v$, where $1 \leq l \leq k-1$. Therefore, the lemma holds. \qed
\end{proof}

\begin{remark}
    \label{remark:follow_the_predecessor_tree_singledoor}
    If $r_1$ and $r_2$ are two successive robots, with $r_1$ entering the graph before $r_2$, then $r_2$ follows the path of $r_1$ until $r_1$ terminates.
\end{remark}

\begin{lemma}
\label{lem:chain_like_structure}
    All the active robots lie on a path rooted at the door.
\end{lemma}
 
\begin{proof}
    Let $r_k$ be the robot that lies on the door, and $r_1$ be the robot that is farthest from the door with $r_1.color, r_k.color \neq$ \texttt{FINISH}. 
    We prove that every robot $r'$ with $r'.color \neq $ \texttt{FINISH} lies on the path $P$ joining the current vertex of $r_1$ and the door. On the contrary, let's assume that $r$ is the robot that does not lie on $P$. From the Lemma \ref{lem:lie_same_subtree}, $r$ and $r_1$ both lie on the same subtree of the door. Let $v$ be the farthest vertex from the door that lies on the intersection of $P$ and the path between the door and $r$.
    Since $G$ is a tree, both $r$ and $r_1$ must leave the vertex $v$ at some time of their traversal. By Lemma \ref{lem:lie_same_subtree}, $r$ and $r_1$ must lie on the same subtree of the vertex $v$, which is a contradiction. \qed
    \end{proof}

\begin{corollary}
\label{cor:unique_follower}
    A robot $r$ with $r.color = \texttt{OFF-i}$ can see at most one robot with color $\texttt{OFF-(i+1)(mod 3)}$ and at most one robot with color $\texttt{OFF-(i-1)(mod 3)}$, for $\texttt{i} \in \{0, 1, 2\}$.
\end{corollary}

\begin{proof}
        By Lemma \ref{lem:chain_like_structure}, all the visible robots of color \texttt{OFF-(i+1) (mod 3)} (or \texttt{OFF-(i-1) (mod 3)}) along with $r$ are on the same path from the door. 
        From the steps of the algorithm, the robot which left the door before $r$ is with the color \texttt{OFF-(i-1) (mod 3)} and the robot which left the door after $r$ is of color \texttt{OFF-(i+1) (mod 3)}. 
    Thus, $r$ can see at most one such robot, as mentioned in the lemma statement. \qed
\end{proof}

\begin{remark}
    \label{remark:chain_of_robots_tree}

    We have the following observations. 
    
    \begin{enumerate}
        \item Starting from the door, the robots lying on the path having the vertices visited by the current head robot constitute the \emph{chain} of robots.
        By Remark \ref{remark:follow_the_predecessor_tree_singledoor}, any active robot must be a member of the current chain and follow the path traversed by the head. 

        \item The head of the current chain terminates before other members of the chain, and after termination, its immediate follower becomes the new head of the chain.
    \end{enumerate}
    
\end{remark}

\begin{lemma}
    \label{lemma:max_min_distance-btw-robots}
    The distance between two consecutive robots (except the one in the door) in the current chain of robots is at most three and at least two.
\end{lemma}

\begin{proof}
    By Corollary \ref{cor:unique_follower}, a robot $r$ with color \texttt{OFF-i} can see at most two robots on the current chain: one $r'$ with the color \texttt{OFF-(i-1) (mod 3)} which it follows, the other $r''$ with the color \texttt{OFF-(i+1) (mod 3)} which follows $r$. 
    Thus, whenever a robot with the color \texttt{OFF-(i+1) (mod 3)} comes within the $2$ hop distance from $r$, it correctly determines the immediate follower robot on the current chain and executes a movement when the follower is at $2$ hops away but does not find any robot with \texttt{OFF-(i-1) (mod 3)} within its $2$ hops visibility. 
    Till the time the follower $r''$ does not reach within $2$ hops of visibility, $r$ does not move ahead.
    The movement of $r$ makes $distance(r'', r) = 3$ and $distance(r, r') = 2$. \qed
\end{proof}

\begin{lemma}
    \label{lem:not_visit_old_edges}
    No robot visits a vertex from its current location, which is already visited by itself in one of its previous LCM cycles.
\end{lemma}

\begin{proof}
    During the computation of the eligible port, a robot $r$ situated at the vertex $v$ disregards the port (incident to $v$) that leads to the edge $e(v, v')$, where $v'$ is the vertex on the path connecting $r$ and its follower on the current chain. 
    Hence, $r$ cannot move in the direction of $v'$ from $v$. 
    Additionally, since $G$ is a tree and hence acyclic, $r$ cannot visit a vertex that is already been visited. \qed
    
\end{proof}

\begin{remark}
    \label{rem:collision-free}
    By Lemma \ref{lemma:max_min_distance-btw-robots}, the distance between two successive robots on the current chain is at least two and at most three.
    By Lemma \ref{lem:not_visit_old_edges}, if a robot $r$ reaches the vertex $v$ from $v'$, it never returns to $v'$.
    Both of these lemmas lead to collision-free movement of a robot. 
\end{remark}

\begin{lemma}
\label{lem:time_asyncSS}
    The algorithm terminates in $O(|E|)$ epochs.
\end{lemma}

\begin{proof}
    Notice that the head of the current chain of robots computes the minimum eligible port and moves forward only if the immediate follower is at $2$ hops away from it. 
    Before proceeding further in the proof, we define a set of \emph{explored edges} by the head. 
    The set of \emph{explored edges} includes those for which at least one endpoint is visited by the head for the first time.
    We first prove the following claim.

    \noindent\emph{Claim: A robot requires $O(k)$ epochs to explore $k$ new edges (if any) starting from the moment it becomes the head of the current chain until termination.}\\
    We begin by proving the claim for the first head (i.e., the first robot that entered the graph).
    Let us denote the first head by $r_1$.
    We discuss the base cases for $k = 2$ and $k = 3$ to give a perspective to the reader.
    It takes $1$ epoch to move to the neighbour of the door from the door along the minimum eligible port incident to the door with the color \texttt{OFF-0}. 
    It takes at most another 2 epochs to reach a 2-hop neighbour of the door: one for $r_2$ at the door to switch to \texttt{OFF-1}, and another one for $r_1$ to move to a 2-hop neighbour upon seeing the follower $r_2$ with \texttt{OFF-1}. 
    Therefore, the head $r_1$ takes $1$ epoch to explore $k = 2$ edges, $3$ epochs to explore $k = 3$ edges and another epoch is required if $r_1$ turns its color to \texttt{FINISH} and terminates at the current location.
    In general, at the end of $t$-th epoch, we assume that the head $r_1$ explores $l$ edges of $G$.
    Let $\{ r_1, r_2, \cdots, r_p\}$ be an ordered set of all robots lying on the current chain. 
    By Lemma \ref{lemma:max_min_distance-btw-robots}, $2 \leq distance(r_i, r_{i+1}) \leq 3$ (except at the door), where $1 \leq i \leq p-1$. 
    Depending on the waiting time of the head $r_1$ between two consecutive movements, we define two types of waiting: \emph{short wait} and \emph{long wait}. Note that we use such conventions strictly from a global view of the configuration and not from the point of view of a particular robot. 

    \begin{itemize}
        \item We call a waiting time of $r_1$ the \emph{short wait}, if between two movements, either $r_1$ does not wait for any other robots (i.e., $distance(r_1, r_2) = 2$) or waits only for the immediate next robot $r_2$ on the chain (i.e., $distance(r_1, r_2) = 3$ and $distance(r_2, r_3) = 2$).

        \item Otherwise, we call it the \emph{long wait} between two movements, i.e., there exist a $r_q$ for $3 \leq q \leq p$ such that $distance(r_1,r_2) = distance(r_2, r_3) = \cdots = distance(r_{q-1}, r_q) = 3$, and either $r_q$ is at the door or $distance(r_q, r_{q+1}) = 2$.
    \end{itemize}
    If $distance(r_1, r_2) = 2$, the head $r_1$ explores one more edge in another epoch. 
    If instead, $distance(r_1, r_2) =3$ and $distance(r_2, r_3) = 2$, then $r_1$ requires at most $2$ epochs to explore one more edge, where one epoch is needed for $r_2$ to move at a vertex so that $distance(r_1, r_2)$ becomes $2$ and another epoch for $r_1$ to execute the movement. 
    Therefore, in case of short wait, the current head $r_1$ takes at most $2$ epochs to explore a new edge after another movement. 
    A more critical scenario for $r_1$ is in the case of the long wait, where $r_1$ has to wait not only for $r_2$, but for $r_3$ to execute a movement towards $r_2$ so that $distance(r_2, r_3) = 2$, leading $r_2$ to execute a movement towards $r_1$ to make $distance(r_1, r_2) = 2$.
    In general, every $r_i$ waits for its follower $r_{i+1}$ on the current chain ($1 \leq i\leq q-1$) to move at a vertex $2$ hops away from it.
    Without loss of generality, we assume that $r_q$ does not wait for a complete epoch for its follower $r_{q+1}$ on the current chain.
    This is not a restrictive assumption for us in this process of the proof.
    Under the $\mathcal{ASYNC}$ scheduler, the adversary can activate $r_1$ first, followed by $r_2$ and so on till $r_q$ within an epoch. 
    Thus, in the worst case, $r_1$ needs to wait for $q-1$ epochs since its last movement, before executing another movement to cover a new edge. 
    However, we show that the robots $r_3, r_4, \cdots r_q$ execute multiple movements during these $q-1$ epochs when $r_1$ is waiting without any movement.
    Let $t$ be the epoch when $r_1$ has executed the last movement.
    Notice that, at the $(t+1)$-th epoch, $r_q$ moves toward $r_{q-1}$ to make $distance(r_{q-1}, r_q) = 2$.
    Similarly, at the $(t+2)$-th epoch, $r_{q-1}$ moves toward $r_{q-2}$. This movement makes $distance(r_{q-1}, r_{q-2}) = 2$ and $distance(r_{q-1}, r_q) = 3$. 
    At $(t+3)$-th epoch, both $r_{q-2}$ and $r_q$ move toward $r_{q-3}$ and $r_{q-1}$ respectively such that $distance(r_{q-1}, r_q) = 2$, $distance(r_{q-2}, r_{q-1}) = 3$ and $distance(r_{q-3}, r_{q-2}) = 2$.
    Similarly, at $(t+4)$-th epoch, both $r_{q-3}$ and $r_{q-1}$ move toward $r_{q-4}$ and $r_{q-2}$ respectively.
    In general, at $(t+(2j-1))$-th epoch, all of $r_{q}, r_{q-2}, \cdots, r_{q-(2j-2)}$ move together, where $1 \leq j \leq \frac{q}{2}$. Additionally, at $(t+2j)$-th epoch, all of $r_{q-1}, r_{q-3}, \cdots,$ $ r_{q-(2j-1)}$ move together, where $1 \leq j < \frac{q}{2}$.
    Finally, at the $(t+q)$-th epoch, $r_1$ moves forward upon seeing $r_2$ at $2$ hops away from its current location.
    At the end of $(t+q)$-th epoch, we have $distance(r_1, r_2)=3, distance(r_2, r_3) = 2, distance(r_3, r_4) = 3, \cdots, distance(r_{q-1}, r_q) = 2 \text{ or } 3$ (depending on odd or even $q$).
    Let us assume that $t' (> t+q)$ is the latest epoch, where the head $r_1$ again encounters a long wait.
    This means that there exists a $r_s$ for $s \geq 3$ such that $distance(r_1,r_2) = distance(r_2, r_3) = \cdots = distance(r_{s-1}, r_s) = 3$, and either $r_s$ is at the door or $distance(r_s, r_{s+1}) = 2$.
    We segregate two sub-cases as follows, based on the relation between $s$ and $q$.
    \begin{itemize}
        \item When $s \geq q$, observe that, at the $(t+q)$-th epoch,  $distance(r_1, r_q) \leq \frac{5(q-1)}{2}$ and at the epoch $t'$, $distance(r_1, r_q) = 3(q-1)$.
        Therefore, $r_1$ has explored at least $3(q-1) - \frac{5(q-1)}{2} = \frac{q-1}{2}$ new edges in between the epochs $(t+q)$ and $t'$.
        Moreover, by the definition of $t'$, the head $r_1$ can encounter only short waits between the two epochs $(t+q)$ and $t'$.
        Consequently, $r_1$ explores $\frac{q-1}{2}$ new edges in at most $2 \frac{(q-1)}{2} = q-1$ epochs, leading to $t' \leq t+q+q-1 = t+2q-1$.
        Hence, the head $r_1$ explores $1+\frac{(q-1)}{2} = \frac{q+1}{2}$ new edges in between the two epochs $t$ and $t'$, i.e., in at most $(2q-1)$ epochs. 
        Therefore, $r_1$ explores a new edge in an average of at most $\frac{(2q-1)}{(q+1)/2} \leq 4$ epochs.

        \item When $s < q$, we have, at the $(t+q)$-th epoch,  $distance(r_1, r_s) \leq \frac{5(s-1)}{2}$ and at the epoch $t'$, $distance(r_1, r_q) = 3(s-1)$.
        Therefore, $r_1$ has explored at least $3(s-1) - \frac{5(s-1)}{2} = \frac{s-1}{2}$ new edges in between the epochs $(t+q)$ and $t'$, where $r_1$ can only encounter short waits.
        Consequently, $r_1$ explores $\frac{s-1}{2}$ new edges in at most $2 \frac{(s-1)}{2} = s-1$ epochs, leading to $t' \leq t+q+s-1$.
        Since $r_1$ encounters a long wait at the $t'$-th epoch, it needs at most $s$ more epochs to explore a new edge.
        Hence, the head $r_1$ explores $1+\frac{(s-1)}{2} = \frac{s+1}{2}$ new edges in between the two epochs $t+q$ and $t'+s$, i.e., in at most $(t'+s)-(t+q) \leq t+q+s-1+s-(t+q) = (2s-1)$ epochs. 
        Therefore, in this case, $r_1$ explores a new edge in an average of at most $\frac{(2s-1)}{(s+1)/2} \leq 4$ epochs.
    \end{itemize}
In both accounts, we can say that either a long wait for $q$ epochs of the head $r_1$ is followed by a sequence of $\frac{q-1}{2}$ short waits (first case) or a sequence of $\frac{s-1}{2}$ short waits of $r_1$ is followed by a long wait for $s$ epochs (second case), leading to an average constant waiting time in between two movements of the head. 
By Lemma \ref{lem:not_visit_old_edges}, every movement of the current head $r_1$ leads it to a new explored edge of $G$. 
Lemma \ref{lem:not_visit_old_edges} and the above arguments jointly conclude the claim for the first head $r_1$.

We extend the proof of the claim for any current head, say $r_i$, except $r_1$.
By Remark \ref{remark:chain_of_robots_tree}, $r_i$ becomes the head of the current chain at the $t_i$-th epoch after all the robots in $\{r_1, r_2, \cdots r_{i-1} \}$ terminate.
By Remark \ref{remark:follow_the_predecessor_tree_singledoor}, $r_i$ follows the path traversed by $r_{i-1}$ before the $t_i$-th epoch. 
Let at the $t_i$-th epoch, $r_{i-1}$ and $r_i$ situated at the vertices $v_{i-1}$ and $v_i$ such that $distance(v_{i-1}, v_i) = 2$ and the vertex lying between $v_{i-1}$ and $v_i$ is $v'$.
Notice that all the edges incident to $v_i$ and $v'$ are already explored by the robots $r_1, r_2, \cdots r_{i-1}$. 
In the worst case, $r_i$ reaches to the vertex $v'$ to cover the edges (if any) incident to $v'$ except the edge $e(v_{i-1}, v')$ (which is already covered by $r_{i-1}$).
When $r_i$ reaches $v'$, the definition of the eligible port prevents $r_i$ to choose the port incident to $v'$ leading to the edge $e(v', v_{i-1})$. 
Lemma \ref{lem:not_visit_old_edges} shows that $r_i$ also does not choose the port leading the edge $e(v', v_{i})$.
Therefore, any movement from $v'$ leads $r_i$ toward a new edge to explore it. 
The same argument ensures that $r_i$ finds new explored edges after every future movement from $v'$.
Starting from the moment when $r_i$ becomes the head of the current chain of robots until termination, if $r_i$ explores $k_i$ edges, it needs at most $k_i + 1$ movements.
The similar arguments presented in the proof of the claim for $r_1$, the robot $r_i$ needs $O(k_i+1) \approx O(k_i)$ epochs. Hence the claim holds true for any robot $r_i$.

Recall that the set of \emph{explored edges} for a robot $r_i$ includes those for which at least one endpoint is visited by $r_i$ for the first time.
It is possible for $r_i$ to terminate without exploring any new edge. This case of trivial as in such a case, it needs a constant many epochs $c_i$ to terminate.
In the worst case, it is evident that the robots need to explore $O(|E|)$ many edges of $G$. 
Let $r_1, r_2, \cdots, r_N$ be the robots with the color \texttt{FINISH} at termination, where the robots $r_i$ explores $k_i$ many edges. 
Then we have $\sum_{i=1}^N k_i = O(|E|)$.
The above claim proves that the total number of epochs needed by all the robots to reach termination is $\sum_{i=1}^N (c_i+ O(k_i)) = O(N) + O(|E|) = O(|E|)$. \qed
\end{proof}

\begin{lemma}
    \label{lem:feasibility_asyncSS}
    At termination, all the \texttt{FINISH}-colored robots form an MVC of $G$.
\end{lemma}

\begin{proof}
    We divide the proof into two parts. First, we prove the feasibility, i.e., every edge is covered by at least one \texttt{FINISH}-colored robot and then we prove the minimality.

    \noindent \textit{Feasibility}: We prove by contradiction. Let $e = e(u, v)$ be the edge which is not covered by any \texttt{FINISH}-colored robot. $e$ can not be the edge incident to the door, as the termination condition suggests. 
    W.l.o.g, let $distance(u, door)< distance(v, door)$ and one of the neighbours of $u$ is occupied with a \texttt{FINISH}-colored robot. 
    Let $r$ be the last robot to terminate (with the color \texttt{FINISH}) among all the robots occupying the neighbours of $u$.
    When $r$ has decided to set its color to \texttt{FINISH} and terminate, it must have seen a \texttt{FINISH}-colored robot either on $u$ or on $v$, which is a contradiction. 

    \noindent \textit{Minimality:} Let the vertex $v_1$, which is occupied by a \texttt{FINISH}-colored robot $r$, be removed from the vertex cover set, and we still have a feasible vertex cover of $G$. 
    So, all the neighbours of $v_1$ must be occupied by \texttt{FINISH}-colored robots. 
    $v_1$ can not be the door, as a robot on the door sets its color to \texttt{FINISH} when it finds at least one unoccupied neighbour. 
    Also, the design of our algorithm ensures that $v_1$ can not be a leaf vertex. 
    Let $v_2$ be the neighbour of $v_1$ such that $distance(v_2, door) < distance(v_1, door)$.
    Then all the robots except on $v_2$, which are positioned on the neighbours of $v_1$, must cross the vertex $v_2$ and then $v_1$ to reach their final position. 
    So, $r$ sets its color to \texttt{FINISH} after all its neighbours except on $v_2$ set their color to \texttt{FINISH}, by Remark \ref{remark:follow_the_predecessor_tree_singledoor} 
    Considering the time $t$ when $r$ is at $v_2$ and $r$ decides to move at $v_1$, it must have found an unoccupied neighbour of $v_1$. It leads to a contradiction, as all the neighbours of $v_1$ at $t$ are already occupied. \qed
\end{proof}

\begin{lemma}
    \label{lem:optimal_tree}
    At termination, all the \texttt{FINISH}-colored robots form a MinVC of $G$.
\end{lemma}

\begin{figure}[h]
 \centering
     \includegraphics[width=0.7\linewidth]{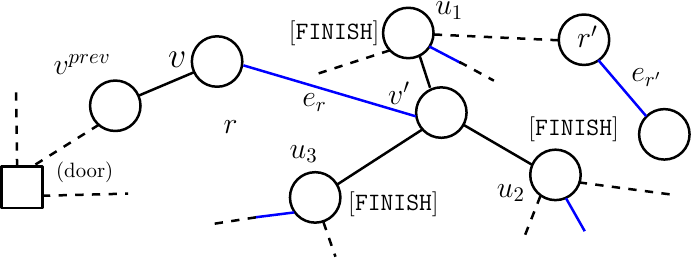}
     \caption{The blue edges are the edges that are covered by only one \texttt{FINISH}-colored robot}
     \label{fig:optimal_tree}
\end{figure}

\begin{proof}
    Let $r$ be a robot on the vertex $v$ that sets its color to \texttt{FINISH} at the time $t$. 
    There can be two cases depending on the position of  $v$. 
    Firstly, we discuss when $v$ is not the door vertex. Then $r$ must reach $v$ from some other vertex $v^{prev}$. Consider the time $t' (< t)$ when $r$ situated at $v^{prev}$ decides to move at $v$. 
    At $t'$, $r$ must find a $2$ hops unoccupied neighbour, say $v'$, through the port $p_{v^{prev}}^{min}$ incident to $v^{prev}$. 
    Let $p$ be the port of the edge $e(v^{prev}, v)$ incident to $v$.
    Since $r$ sets its color to \texttt{FINISH} lying at $v$, all the ports incident to $v$ are ineligible.
    This implies that all the neighbours of $v'$ are occupied.
    By Remark \ref{remark:chain_of_robots_tree}, all neighbours of $v'$ have color \texttt{FINISH} at $t$, as they left the door prior to $r$.
    After the time $t$, any robot on a vertex $u_i$, which is on the same subtree rooted at $v^{prev}$ as $v$ and $distance(u_i, v^{prev}) > distance(v, v^{prev})$ (refer to Fig. \ref{fig:optimal_tree}) can not reach $v'$, as it is with the color \texttt{FINISH}. 
    Any other robot also cannot reach $v'$ after $t$, as it has to cross the vertex $v$ to reach $v'$, which is not possible as $v$ is occupied. 
    Hence, there is an edge $e(v, v')$ that is covered only by the robot $r$.
    We denote this edge by $e_r$.
    Consequently, all the incident edges of $v$ and $v'$ are covered.
   Let us hypothetically remove all such edges from $G$, as they are already covered.
   Consider another robot $r'$ which has set its color to \texttt{FINISH} after $r$. 
   The edge $e_{r'}$, covered only by $r'$, is neither incident to $v$ nor $v'$.
   When $v$ is the door, an edge $e(v, v')$ must exist that is not covered by any \texttt{FINISH}-colored robot before the time $t$.
   Let $e(v, v') = e_r$.
   Hence, for every two \texttt{FINISH}-colored robots $r$ and $r'$, we can always select two edges $e_r$ and $e_{r'}$ such that they do not share any common vertex.
   Let $E' \subset E$ be all such edges. 
    $E'$ must be a matching (independent set of edges) of $G$, and $|E'|$ equals the total number of \texttt{FINISH}-colored robots at the end of the algorithm.
    Let $OPT$ be the size of the minimum vertex cover of $G$. 
    Since any two edges in $E'$ do not share a vertex, any vertex cover must have at least $|E'|$ many vertices to cover all the edges of $E'$. 
    So, $OPT \geq |E'|$. $OPT \leq |E'|$, as $E'$ form an MVC of $G$, by Lemma \ref{lem:feasibility_asyncSS}. Hence, $OPT = |E'|$. \qed
\end{proof}

\noindent In the algorithm \textsc{Tree\_SingleDoor}, we use $4$ colors in total: \texttt{OFF-0}, \texttt{OFF-1}, \texttt{OFF-2} and \texttt{FINISH}.
Combining Lemma \ref{lem:time_asyncSS}, \ref{lem:feasibility_asyncSS}, \ref{lem:optimal_tree}, and Remark \ref{rem:collision-free}, we conclude the following theorem.



\begin{theorem}
    \label{theorem:tree-singledoor}
    Algorithm \textsc{Tree\_SingleDoor} fills a MinVC of a tree $G$ with \texttt{FINISH}-colored robots having $2$ hops visibility and $4$ colors in $O(|E|)$ epochs under $\mathcal{ASYNC}$ and with no collision.
\end{theorem}

\subsubsection{Execution Example of the Algorithm \textsc{\textmd{Tree\_SingleDoor}.}}
\label{subsec-app:execution-example}

 The example provides a foundational understanding of the technique, which extends naturally to the multiple-door setting.
Please refer to Fig. \ref{fig:example} and the following description.

\begin{figure}[h]
    \centering
     \includegraphics[width=0.7\linewidth]{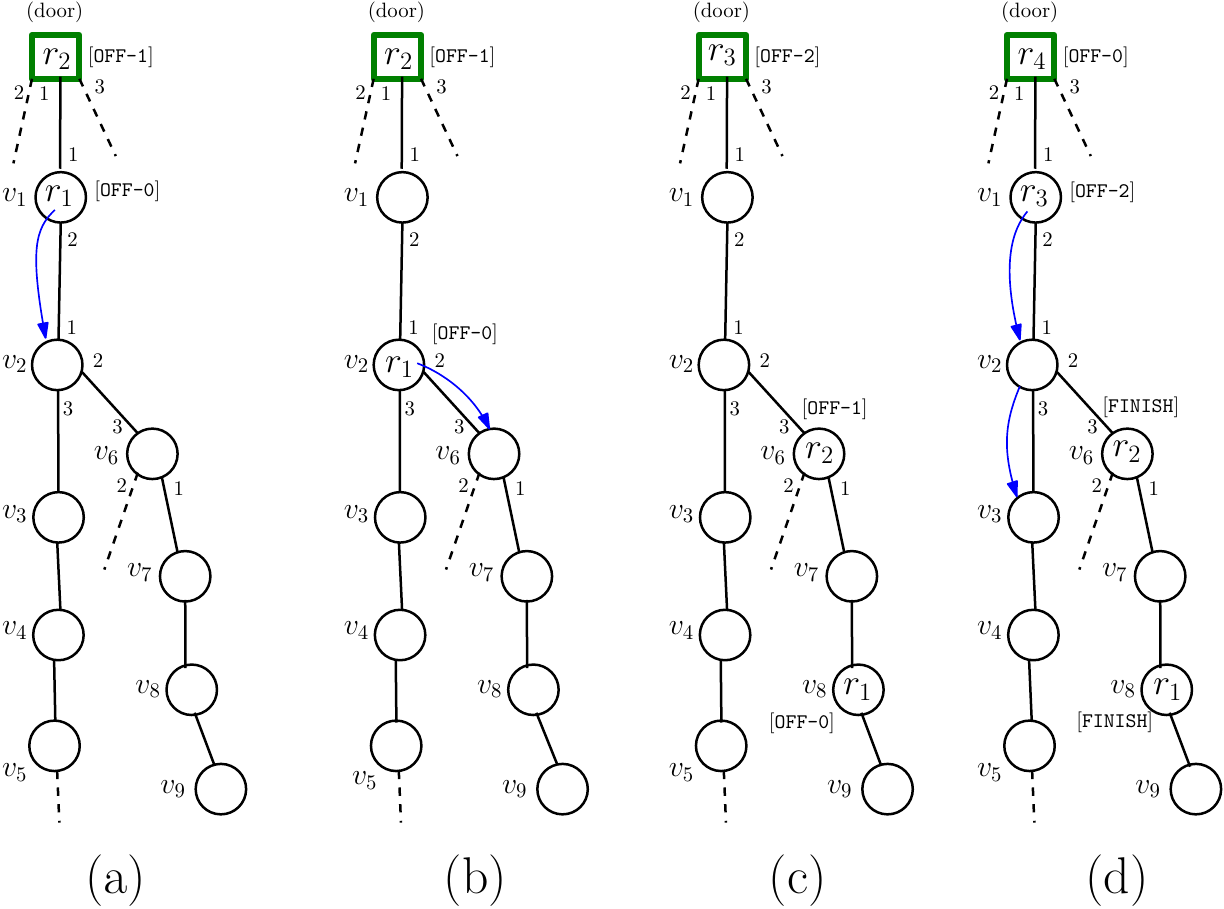}
     \caption{An execution of the algorithm \textsc{Tree\_SingleDoor}}
     \label{fig:example}
\end{figure}
Initially, the robot $r_1$ appears at the door (which is indicated with a green square in Fig. \ref{fig:example}) with the color \texttt{OFF-0}. 
After finding all the $1$ and $2$ hops neighbours unoccupied, it  calculates $p_{door}^{min} = 1$ and moves to $v_1$ with \texttt{OFF-0}.
Another robot $r_2$ is placed at the door after $r_1$ leaves.
$r_1$ remains stationary till it finds an \texttt{OFF-1}-colored robot.
$r_2$ switches its color to \texttt{OFF-1}, when it sees $r_1$ with color \texttt{OFF-0} at one of its neighbour $v_1$.
Now, $r_1$ discards the port $1$ as the eligible port after seeing $r_2$ along the port $1$ and computes $p_{v_1}^{min} = 2$. 
$r_1$ moves to $v_2$ with \texttt{OFF-0}, as shown in Fig. \ref{fig:example}(a).
Thereafter, $r_1$ moves to the neighbour $v_6$ of $v_2$ along the port $p_{v_2}^{min} = 2$ with \texttt{OFF-0}, while $r_2$ maintains status quo as it finds $r_1$ sitting at its $2$ hops neighbour with color $\neq \texttt{FINISH}$ (refer to \ref{fig:example}(b)).
In the same way, $r_1$ and $r_2$ eventually move to $v_8$ and $v_6$ respectively (Fig. \ref{fig:example}(c)). 
Meanwhile, another robot $r_3$ on the door sets its color \texttt{OFF-2}.
Now, $r_1$ terminates at $v_8$ with color \texttt{FINISH}, as it finds $r_2$ at $v_6$ and no other $2$ hops unoccupied neighbour.
$r_2$ similarly changes its color to \texttt{FINISH}, when $r_3$ is at $v_1$, as depicted in Fig. \ref{fig:example}(d).
Next, $r_3$ moves to $v_2$ and then calculate $p_{v_2}^{min} = 3$ after disregarding the ports $1$ and $2$.
Finally, $r_3$ moves to $v_3$ and continues the filling process.

\subsection{Algorithm (\textsc{\textmd{Tree\_MultiDoor}}) for Multiple Doors}
\label{subsec:tree-multidoor}
We consider $H (>1)$ doors, each assigned a unique color (can be thought of as ID) to establish a hierarchy, where the lower-ID doors dominate the higher-ID ones.
Barrameda et al. \cite{10.1007/978-3-540-92862-1_11} showed that $\Omega(H)$ colors are necessary to distinguish the robots entering from different doors. 
In this algorithm, robots use $4$ hops visibility and use $O(H)$ colors.
Initially, robots at the $h$-th lowest ID door are colored $\texttt{color}_h^0$, and two doors are not adjacent to each other.
Due to space constraints, we could not include the detailed description of the algorithm, and hence, we present a high-level idea of it.

\noindent \textbf{High-level Idea:} 
The strategy mirrors \textsc{Tree\_SingleDoor}, keeping a chain-like formation from each door, but with a distinct set of colors. Robots from $h$-th lowest ID door use the color set $\{\texttt{color}_h^0, \texttt{color}_h^1, \texttt{color}_h^2\}$, similar to $\{\texttt{OFF-0},$ $ \texttt{OFF-1}, \texttt{OFF-2}\}$ in previous algorithm.
The head of each chain explores the graph, while the other robots follow it. 
A head $r$ at vertex $v$ with the color \texttt{color}$_h^i$, moves to $v_{nbr}^1$ along the minimum eligible port to cover the edge $e(v_{nbr}^1, v_{nbr}^2)$, if $v_{nbr}^2$ is unoccupied. 
After moving, $r$ remains stationary until it finds its follower with \texttt{color}$_h^{i+1}$ at a distance of 2 hops. 
If no eligible port is available, $r$ terminates with \texttt{FINISH}, and the follower of $r$ becomes the new head.

\begin{figure}[h]

\centering
\begin{minipage}[b]{0.48\linewidth}
\centering
     \includegraphics[width=0.9\linewidth]{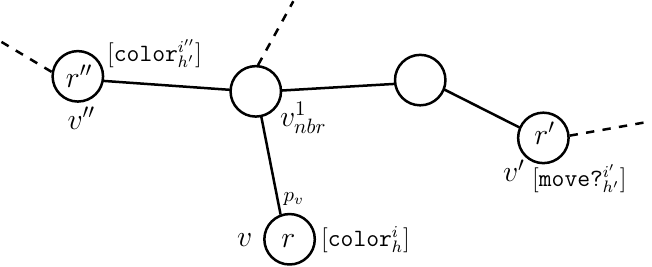}
     \caption{$r$ sees $r'$ and $r''$ due to $3$ hop visibility and do not move to $v_{nbr}^1$ to avoid possible collision}
     \label{fig:never_cross_chain}
\end{minipage}
\hfill
\begin{minipage}[b]{0.48\linewidth}
\centering
         \includegraphics[width=0.85\linewidth]{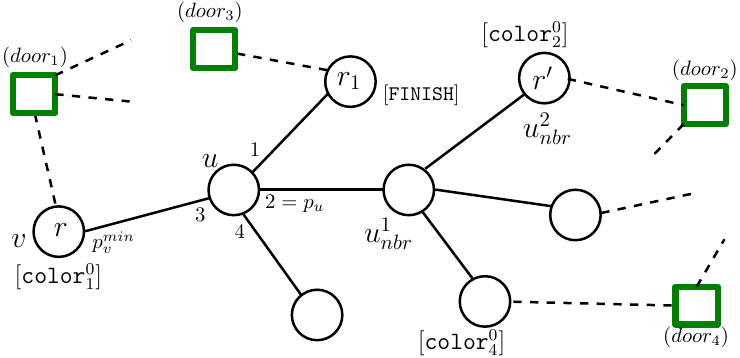}
         \caption{$r'$ moves to $u_{nbr}^1$ while $r$ remains stationary in spite of $r$ from higher ID door.}
         \label{fig:async_critical_multiDoor}
\end{minipage}
\end{figure}

To ensure collision-free operation, robots perform additional checks beyond maintaining the chain.
First, to prevent chain intersections, a robot $r$ considers a port $p_v$ \emph{ineligible} if it detects two robots $r'$ and $r''$ (from the same door but different from $r$) such that $v_{nbr}^1$ along $p_v$ lies on the path between their positions $v'$ and $v''$ (see Fig.~\ref{fig:never_cross_chain}).
Second, if two head robots from different doors potentially target the same vertex, only the robot from the lower ID door proceeds.

We address a critical $\mathcal{ASYNC}$ scenario to avoid collision. 
Let $r$ at $v$ with \texttt{color}$_h^i$ finds both $v_{nbr}^1 (=u)$ and $v_{nbr}^2 (=u_{nbr}^1)$ unoccupied along its minimum eligible port.
Meanwhile, another robot $r'$ at $u_{nbr}^2$ with \texttt{color}$_{h'}^{i'} (h' > h)$ sets $u_{nbr}^1$ as its target but hasn't moved yet (see Fig. \ref{fig:async_critical_multiDoor}). 
$r$ being a robot from a lower ID door than $r'$, it might complete two LCM cycles: first moves to $u$ and then to $u_{nbr}^1$, resulting in a collision with $r'$. 
To prevent this, $r$ remains stationary at $v$, allowing the closer robot $r'$ to proceed.
However, if $r'$ is initially outside the view of $r$ (e.g., approaching $u_{nbr}^2$ from $door_2$), $r$ can target $u$ without moving yet.
Once $r'$ reaches $u_{nbr}^2$ and sees $r$ at $v$, it may incorrectly assume itself as a robot closer to $u_{nbr}^1$ than $r$ and sets $u_{nbr}^1$ as the target, leading to  another collision.
To resolve this, if $r'$ first switches to \texttt{move?}$_{h'}^{i'}$ to signal its intent to move to $u_{nbr}^1$ and waits for confirmation from $r$.
In case $r$ reaches $u$, $r'$ in its next LCM cycle, after activating with \texttt{move?}$_{h'}^{i'}$, gives up its plan to move $u_{nbr}^1$ and reverts to \texttt{color}$_{h'}^{i'}$. 
If instead, $r$ is still at $v$, $r'$ moves to $u_{nbr}^1$ with \texttt{color}$_{h'}^{i'}$ safely.

\vspace{2em}
\noindent \textbf{Description of Algorithm \textmd{\textsc{Tree\_MultiDoor:}}}
Since $r$ performs different actions based on which predicate is satisfied, we use \textsc{If, ElseIf} (as in pseudocode), meaning the predicate in the $i$-th \textsc{ElseIf} is evaluated only if all previous ones are false.
For this algorithm, we modify the definition of an eligible port as follows.

\begin{definition}{(\textbf{Eligible Port from $v$})}
A port $p_v$ incident to $v$ is called eligible for $r$, positioned at $v$ with $r.color \in \{\texttt{color}_h^{i}, \texttt{move?}_h^i\}$, if the following conditions hold. 
(i) There exists an unoccupied $v_{nbr}^2$ along the port $p_v$.
(ii) No robot $r'$ with $r'.color = $ \texttt{color}$_h^{i'}$, for $1 \leq h \leq H$ and $0 \leq  i \neq i' \leq 2$ is present at a $v_{nbr}^2$ along $p_v$. 
(iii) $r$ doesn't see two robots $r'$ at $v'$ and $r''$ at $v''$ such that $v_{nbr}^1$ or $v_{nbr}^2$ along the port $p_v$ lies on the path between $v'$ and $v''$ with $r'.color =$ \texttt{col}$_{h'}^{i'}$ and $r''.color =$ \texttt{col}$_{h'}^{i''}$, where \texttt{col} $\in \{\texttt{color}, \texttt{move?} \}$, $i', i'' \in \{0, 1, 2\}$, and $h' \neq h$. (refer to Fig. \ref{fig:never_cross_chain})

\end{definition}

We identify following three cases based on the position and color of $r$.

\noindent $\blacktriangleright$~  \textbf{Case B.1 ($r$ is at the $door_h$):} We further identify the two sub-cases based on the current color of $r$.
    \begin{itemize}
        \item \textbf{Case B.1.1} ($r.color =$ \texttt{color}$_h^0$): \textsc{If} $\exists p_v ~ occupied(v_{nbr}^1,$ $ p_v, \texttt{color}_h^i)$ or $\exists p_v \exists v_{nbr}^2$ $ occupied(v_{nbr}^2, p_v, \texttt{color}_h^i)$ for $\texttt{i}\in \{0,1,2\}$, it changes its color to \texttt{color}$_h^j$ where \texttt{j = i+1 (mod3)}. 
        \textsc{ElseIf}, $ \exists p_v \exists v_{nbr}^3~ occupied(v_{nbr}^3, p_v, \texttt{color}_h^2)$, $r$ moves to the vertex $v_{nbr}^1$ along the port $p_v$ with the current color.

        \item \textbf{Case B.1.2} ($r.color =$ \texttt{color}$_h^i$ for \texttt{i} = $1$ or $2$):  \textsc{If} $\exists p_v~ occupied(v_{nbr}^1, p_v, \neq \texttt{FINISH}) \lor \exists v_{nbr}^2 ~ occupied(v_{nbr}^2, p_v, \neq \texttt{FINISH}))$, $r$ maintains the status quo.
        \textsc{ElseIf} $ \exists p_v \exists v_{nbr}^3~ occupied(v_{nbr}^3,$ $ p_v, \texttt{color}_h^{i-1})$, $r$ moves to the vertex $v_{nbr}^1$ along the port $p_v$ with the current color. 
        This essentially means that $r$ with color \texttt{color}$_h^i$ follows its predecessor whose color is \texttt{color}$_h^{(i-1)}$.
    \end{itemize}    

    The above two sub-cases deal with all those robots that are followers of some robots from the same door. 
    We now discuss the case when $r$ is the first robot to move through one of the incident ports.
    If the predicates mentioned in the above two sub-cases are not true, irrespective of $r$ having color \texttt{color}$_h^i$ for \texttt{i} $\in \{0,1,2 \}$, it finds the minimum eligible port $p_v^{min}$ (if exists). 
    If $p_v^{min}$ exists, there must be an unoccupied $v_{nbr}^2$ along the port $p_{v}^{min}$. Let us denote $v_{nbr}^1$ along the port $p_{v}^{min}$ by $u$.
    Let $p_u$ be the minimum among all incident ports to the vertex $u$ along which the $1$ hop neighbour of $u$ is unoccupied.

    \begin{itemize}
        \item \textsc{If} $\exists u_{nbr}^3~ occupied(u_{nbr}^3, p_u, \texttt{col}_{h'}^{i'})$ and $h > h'$, $r$ remains in place. 
        If we denote the robot occupying a $3$ hop neighbour of $u$ by $r'$ which has entered the graph through $door_{h'}$, the two robot $r$ and $r'$ can move simultaneously and reach at the vertices $u$ and $u_{nbr}^2$ respectively. 
        This might lead to a violation of the minimality of the vertex cover of the graph. 
        To tackle this possibility, we refrain $r$ from reaching $u$ using $4$ hop visibility. 
        
        \item \textsc{ElseIf} $\exists u_{nbr}^1~ occupied(u_{nbr}^1, p_u, \texttt{col}_{h'}^{i'})$, where $h > h'$ and \texttt{col} $\in \{\texttt{color}, \texttt{move?}\}$, $r$ does nothing. In this case, $r$ allows the robot $r'$ situated on $u_{nbr}^1$ to cover the edge $e(u, u_{nbr}^1)$, as $r'$ is from a lower ID door than $r$.

        \item \textsc{ElseIf} $\neg exists(p_u, 2)$ or $\forall u_{nbr}^2  (\neg occupied(u_{nbr}^2, p_u, \sim) \lor occupied(u_{nbr}^2, p_u, $ $ \texttt{FINISH}))$, $r$ moves to $u$ with the current color. This is because $r$ is either the only robot left to cover the edge $e(u, u_{nbr}^1)$ or all the occupied neighbours of $u$ are from a higher door than $r$.

        \item  \textsc{ElseIf} $\exists u_{nbr}^2~ \neg occupied(u_{nbr}^2, p_u, \sim) $ and $ \exists u_{nbr}^2 ~ occupied$ $(u_{nbr}^2,  p_u, \texttt{col}_{h'}^{i'})$ with \texttt{col} $\in \{ \texttt{color}, \texttt{move?}\}$ such that $h$ is lesser than the minimum among all such $h'$, $r$ remains stationary. Let $r'$ be the robot located at $u_{nbr}^2$ with the smallest $h'$. Although $r'$ comes from the higher ID door than $r$, $r$ allows $r'$ to move at $u_{nbr}^1$. Due to asynchrony, it may happen that $r'$ had already decided to move at $u_{nbr}^1$ before $r$ started its own LCM cycle (See Fig \ref{fig:async_critical_multiDoor}). Hence, the movement of $r$ towards $u$ may create a collision with $r'$.

        \item  \textsc{ElseIf} $\exists u_{nbr}^2~ \neg occupied(u_{nbr}^2, p_u, \sim) $ and $ \exists u_{nbr}^2 ~ occupied$ $(u_{nbr}^2,  p_u, \texttt{color}_{h'}^{i'})$ such that $h$ is greater than the minimum among all such $h'$, $r$ changes its color to \texttt{move?}$_h^i$. This is one of the negation conditions of the above. So, $r$ is willing to move. But before executing the movement, it seeks permission from $r'$. Since $r'$ may have already decided to move at $u_{nbr}^1$. In that case, after seeing the robot on $u_{nbr}^1$, $r$ drops its wish to move at $u$. The detailed conditions are mentioned later in Case B.3. 

        \item \textsc{ElseIf} $\forall v_{nbr}^2~(\neg occupied(v_{nbr}^2, p_v^{min}, \sim) \lor (occupied(v_{nbr}^2, $ $ p_{v}^{min}, \texttt{color}_{h'}^{i'}) \land (h < h')))$, the robot $r$ moves to $u$. This is the case when all the above predicates related to the vertex $u$ do not hold, and $r$ is from the lower ID door among all its $2$ hops occupied neighbours along $p_v^{min}$. 

        \item \textsc{ElseIf} $p_v^{min}$ does not exist and $\exists p_v \neg occupied(v^1_{nbr}, p_v, \sim)$, $r$ changes its color to \texttt{FINISH}.

    \end{itemize}

    \noindent $\blacktriangleright$~  \textbf{Case B.2} ($r$ is not on any door and $r.color =$ \texttt{color}$_h^i$ for \texttt{i} $\in \{0,1,2 \}$):
    $r$ remains stationary till it does not see a robot $r'$ with color \texttt{color}$_h^j$      (the follower of $r$) at most $2$ hops away from the current vertex of $r$, where \texttt{j = i+1 (mod 3)}. 
    It also remains stationary if it finds a robot $r'$ with color \texttt{color}$_h^j$, at most $2$ hops away from $v$, where \texttt{j = i-1 (mod 3)}. 
    If $\exists p_v \exists v_{nbr}^3 ~ occupied(v_{nbr}^3, p_v, \texttt{color}_h^j)$  for \texttt{j = i-1 (mod 3)}, $r$ moves to the vertex $v_{nbr}^1$ along the port $p_v$.
    Otherwise, $r$ checks whether there is any eligible port incident to the current vertex $v$ or not. If there does not exist any eligible port, it changes its color to \texttt{FINISH}. 
    If it exists, let $u = v_{nbr}^1$ along the port $p_{v}^{min}$ and $p_u$ be the port minimum among all the incident ports to $u$ along which the $1$ hop neighbour is unoccupied.

    \begin{itemize}
        \item \textsc{If} $\exists u_{nbr}^3~ occupied(u_{nbr}^3, p_u, \texttt{col}_{h'}^{i'})$ and $h > h'$, $r$ remains in place.
    
        \item \textsc{ElseIf} $\exists u_{nbr}^1~ occupied(u_{nbr}^1, p_u, \texttt{move?}_{h'}^{i'})$, $r$ remains stationary. In this case, $r$ is eligible to move to $u$ if $h < h'$. However, before moving, $r$ allows the robot with color $\texttt{move?}_{h'}^{i'}$ to change its color, so that the robot gives up its intent to move to $u_{nbr}^1$. 

        \item \textsc{ElseIf} $\exists u_{nbr}^1~ occupied(u_{nbr}^1, p_u, \texttt{col}_{h'}^{i'})$, where $h > h'$ and \texttt{col} $\in \{\texttt{color}, \texttt{move?}\}$, $r$ remains in place, as it is from the higher ID door than the robot at $u_{nbr}^1$.

        \item \textsc{ElseIf} $\neg exists(p_u, 2)$ $\lor$ $\forall u_{nbr}^2  (\neg occupied(u_{nbr}^2, p_u, \sim) \lor $ $ occupied(u_{nbr}^2, p_u, $ $ \texttt{FINISH}))$, $r$ moves to $u$ with the current color.

        \item \textsc{ElseIf} $\exists u_{nbr}^2~ \neg occupied(u_{nbr}^2, p_u, \sim) $ and $ \exists u_{nbr}^2 ~ occupied$ $(u_{nbr}^2,  p_u, \texttt{col}_{h'}^{i'})$ with $\texttt{col} \in \texttt{color}, \texttt{move?}$ such that $h$ is lesser than the minimum among all such $h'$, $r$ remains stationary. This is similar to one of the subcases discussed in Case B.1.

        \item  \textsc{ElseIf} $\exists u_{nbr}^2~ \neg occupied(u_{nbr}^2, p_u, \sim) $ and $ \exists u_{nbr}^2 ~ occupied$ $(u_{nbr}^2,  p_u, \texttt{color}_{h'}^{i'})$ such that $h$ is greater than the minimum among all such $h'$, $r$ changes its color to \texttt{move?}$_h^i$ to get a confirmation from the robots with color \texttt{color}$_{h'}^{i'}$.

        \item \textsc{ElseIf} $\forall v_{nbr}^2 (\neg occupied(v_{nbr}^2, p_v^{min}, \sim) \lor (occupied(v_{nbr}^2, $ $ p_v^{min}, \texttt{color}_{h'}^{i'}) \land (h < h')))$, $r$ moves to $u$.
    \end{itemize}

    \noindent $\blacktriangleright$~  \textbf{Case B.3} ($r.color =$ \texttt{move?}$_h^i$):
    Let $u = v_{nbr}^1$, which is the $1$ hop neighbour along the port $p_v^{min}$. 
    If all the $1$ hop neighbours of $u$, except $r$ itself, are unoccupied or occupied by robots with color from the set $\{\texttt{color}_{h'}^{i'}, \texttt{FINISH}, $ $ \texttt{move?}_{h'}^{i'}$\}, where $h' > h$, it moves to the vertex $u$ after changing its color to \texttt{color}$_h^i$. 
    On the other hand, when $r$ sees a robot occupying $u_{nbr}^1$ with the color \texttt{color}$_{h'}^{i'}$, where $h' < h$, it changes its color to \texttt{color}$_h^i$ without any movement. 
    For any other cases, $r$ does nothing.

The termination occurs when for every $1\leq h \leq H$, $door_h$ has a \texttt{FINISH}-colored robot or all the $1$ hop neighbours of the door are occupied with \texttt{FINISH}-colored robots.

\noindent \textbf{Analysis of the \textsc{\textmd{Tree\_MultiDoor}} Algorithm:}
In this section, we analyse the time complexity and the correctness of the algorithm \textsc{Tree\_MultiDoor}. 
The concept of the chain and the head of the chain remain the same as Section \ref{subsec:tree-singledoor}.
Notice that if the graph $G$ has $H$ distinct doors, robots form at most $H$ different chains originating from the distinct doors.

\begin{lemma}
    \label{lemma:time-complexity-tree-multidoor}
    The algorithm terminates in $O(|E|)$ epochs.
\end{lemma}

\begin{proof}
    The claim mentioned in Lemma \ref{lem:time_asyncSS} holds true for the head of any chain of robots till it encounters any other robots from another door situated on a vertex along its minimum eligible port. 
    The algorithmic steps are different when robots from two different doors meet within their visibility range. 
    Without loss of generality, let us assume that the current head $r_h$ of the chain originated from $door_h$ is at the vertex $v$ and encounters another robot $r_{h'}$ from $door_{h'}$ along its minimum eligible port $p_v^{min}$. 
    Let us also denote the neighbour of $v$ along the port $p_{v}^{min}$ by $v_{nbr}^1$.
    By the definition of the eligible port, the robot $r_{h'}$ cannot be positioned at the vertex $v_{nbr}^1$.
    Consider a situation where $r_{h'}$ is located at a vertex $1$ hop away from $v_{nbr}^1$ (i.e., at a $v_{nbr}^2$). 
    If $r_h.color = \texttt{color}_h^i$ and $r_{h'}.color = \texttt{color}_{h'}^{i'}$ for some $i, i' \in \{0,1,2\}$, the robot $r_h$ is eligible to move forward when $h < h'$, otherwise $r_{h'}$.
    If $r_h.color = \texttt{color}_h^i$ and $r_{h'}.color = \texttt{move?}_{h'}^{i'}$ for some $i, i' \in \{0,1,2\}$, the robot $r_h$ waits for one epoch for the robot $r_{h'}$ to let it change its color to $\texttt{color}_{h'}^{i'}$. 
    However if $r_h.color = \texttt{move?}_{h}^{i}$, for some $i \in \{0,1,2\}$, the robot $r_h$ either switches to $\texttt{color}_{h}^{i}$ without any movement or moves to the vertex $v_{nbr}^1$ with the color $\texttt{color}_{h}^{i}$.
    Consider another situation where $r_{h'}$ is located at $2$ hops away from $v_{nbr}^1$.
    If $r_{h'}.color = \texttt{move?}_{h'}^{i'}$, the robot $r_{h'}$ is eligible to move forward toward $v_{nbr}^1$.
    If instead, $r_{h'}.color \neq \texttt{move?}_{h'}^{i'}$, in just $2$ epochs, $r_h$ is eligible to move to $v_{nbr}^1$ when $h > h'$; else it is $r_{h'}$.
    In both situations, we can conclude that the robot $r_h$ or $r_{h'}$ is eligible to move forward in just $2$ epochs.
    The above argument, together with the claim stated in Lemma \ref{lem:time_asyncSS}, establishes the following claim, which is a slight modification of that claim. 
    At any time instance, there exists a chain of robots on which a robot requires $O(k)$ epochs to explore $k$ new edges starting from the moment it becomes the head of the current chain until termination.
    Therefore, by Lemma \ref{lem:time_asyncSS}, the total number of epochs required to reach termination is $O(|E|)$ epochs.  \qed
\end{proof}

\begin{theorem}
\label{thm:coll_tree_MS}
    Algorithm \textsc{Tree\_MultiDoor} fills an MVC of a tree having $H$ doors with \texttt{FINISH}-colored robots having $4$ hop visibility and $O(H)$ colors in $O(|E|)$ epochs under $\mathcal{ASYNC}$ and without collision. 
\end{theorem}

\begin{proof}
    Similarly to Lemma \ref{lem:chain_like_structure}, we can prove that all robots that are not with color \texttt{FINISH} and left from the same door lie on a path. In other words, all robots with color from the set $\{\texttt{color}_h^i, \texttt{move?}_h^i~ | 0\leq i \leq 2\}$ lie on a path starting from the door with the $h$-th lowest ID, for some $1 \leq h \leq H$. This result validates the feasibility of the solution.
    Lemma \ref{lemma:time-complexity-tree-multidoor} The proof of time complexity and minimality can be proved by similar arguments as presented in Lemma \ref{lem:time_asyncSS} and Lemma \ref{lem:feasibility_asyncSS}. 

    The algorithm uses $6H+1 = O(H)$ colors in total: \texttt{color}$_h^0$, \texttt{color}$_h^1$, \texttt{color}$_h^2$, \texttt{move?}$_h^0$, \texttt{move?}$_h^1$, \texttt{move?}$_h^2$ and \texttt{FINISH}, where $1\leq h \leq H$.
    
    We now prove that the movements of the robots are free from collisions. 
    On the contrary, let two robots meet a collision at the vertex $z$. 
    If both of them enter the graph through the same door, they do not collide by Remark \ref{rem:collision-free}.
    Otherwise, let $r_h$ and $r_{h'}$ with $h \neq h'$ be the two robots with colors \texttt{color}$_h^i$ and \texttt{color}$_{h'}^{i'}$, respectively collide at the vertex $z$. 
    Let $x$ and $y$ be the previous positions of $r_h$ and $r_{h'}$, respectively, and let $w$ be the previous position of $r_{h'}$ before it lies on $y$.
    Without loss of generality, we assume that $h > h'$. Consider the time $t$ when the robot $r_h$ is at $x$ and decides to move to $z$. We identify the following two cases. (Case I) At $t$, $r_{h'}$ is at $y$ or on the edge $e(y, w)$. In this case, $r_h$ never moves to $z$ as $h > h'$. (Case II) At $t$, $r_{h'}$ is at $w$. Then, before $r_h$ moves to $z$, it first changes its color to \texttt{move?}$_h^i$. We further divide Case II into the following two sub-cases. 
    (Case IIa) $r_{h'}$ has already decided to move to $y$ from the vertex $w$. 
    After the completion of the movement, either $r_{h'}$ finds $r_h$ with the color \texttt{move?}$_h^i$ or \texttt{color}$_h^i$, or $r_h$ is on edge $e(x, z)$. 
    If $r_h$ is on the edge, $r_{h'}$ does not move to $z$. However, $r_{h'}$ sees $r_h$ on the vertex $x$, it means that $r_h$ has already decided not to move to $z$, so the movement of $r_{h'}$ to $z$ does not make the collision with $r_h$. 
    If $r_h$ is still with the color \texttt{move?}$_h^i$, $r_{h'}$ waits until $r_h$ updates to \texttt{color}$_h^i$. 
    (Case IIb) $r_{h'}$ has not decides to move at $y$ from $w$ at the time $t$. 
    When $r_{h'}$ is activated again after the time $t$, it observes $r_h$ at the vertex $x$ and chooses to remain stationary. 
    Hence, for all the cases, $r_h$ and $r_{h'}$ never collide. \qed
\end{proof}

\noindent{\textbf{Lower Bound on Time for Multidoor.}} 
Consider the graph as depicted in Fig. \ref{fig:lower_bound_multidoor}. The length of the path $P$ from $v$ to $door_1$ is $|E|-2H$.
All the edges on $P$ must be covered by the robots entered through $door_1$.
By Theorem \ref{thm:time_lowerBound}, any algorithm needs $\Omega(|E|-H)$ epochs to cover $P$.
This shows that our algorithm is time-optimal when there is a constant number of doors in the graph.

\begin{figure}[h]
     \centering
     \includegraphics[width=0.5\linewidth]{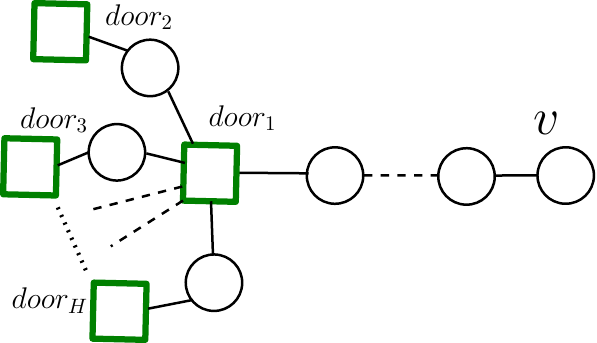}
     \caption{Robots from $door_1$ fill the vertex cover of the path from $door_1$ to $v$}
     \label{fig:lower_bound_multidoor}
 \end{figure}



\section{Algorithm for Filling MVC Vertices on General Graph}
\label{sec:general_graph}

The algorithms designed for trees do not extend directly to general graphs, as multiple paths can exist between two vertices, making it difficult for a robot to accurately determine which robot it is following or which one is following it.
Even if a robot has sufficient persistent memory, relying only on 2-hop visibility may lead to violating the minimality of the vertex cover. 
Robots maintain a \emph{predecessor-successor} relationship during the execution of the algorithm
In general graphs, we assume $3$ hop visibility and $O(\log \Delta)$ additional memory in general graph to store the ports toward the predecessor and the successor. 


\subsection{Algorithm (\textsc{\textmd{Graph\_SingleDoor}}) for Single Door}
\label{subsec:graph-singledoor}
Each robot initially has the color \texttt{OFF} and uses $4$ colors in total. 
We define two variables $r.pred$ and $r.succ$ for a robot $r$ situated at the vertex $v$.

\begin{definition}{\textbf{(Predecessor and Successor of $r$)}}
    $r.pred$ (resp. $r.succ$) is a tuple $(p_1, p_2)$ (resp. $(p'_1, p'_2)$), implying that there is a 2-hop between the predecessor (resp. successor) robot of $r$ situated at $u$ and $r$ situated at $v$ such that $p_1$ (resp. $p'_1$) is the port of the edge $e(v,v_{nbr}^{1})$ and $p_2$ (resp. $p'_2$) is the port of the edge $e(v_{nbr}^1, u)$, where $v_{nbr}^1$ is the neighbour of $v$ along $p_1$ (resp. $p'_1$), 
    see Fig. \ref{fig:predecessor_successor}. 
\end{definition}

\begin{wrapfigure}[7]{r}{0.4\linewidth}

     \centering
    \includegraphics[width=\linewidth]{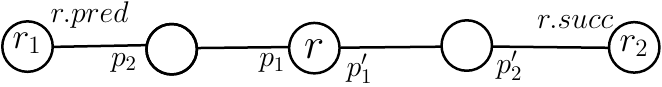}
     \caption{Predecessor and successor of $r$}
     \label{fig:predecessor_successor}
 \end{wrapfigure}
\noindent We misuse the notation $r.pred$ and $r.succ$ to denote the predecessor and successor robot of $r$ as well.
The computation of $r.pred$ is possible when the predecessor is exactly at $2$ hops away from $v$. Otherwise, $r$ can only store the first component of the tuple $r.pred$. In that case, $r.pred = (p_1, \bot)$.
$r$ can also have $r.pred = (\bot, \bot)$, indicating that $r$ does not have a predecessor, and we call such robot $r$ as \emph{head}.
$r$ also maintains a variable $r.avoid$, storing a port $p$ incident to its current position, and all ports $p' < p$ are ignored when computing eligible port.
Initially, $r$ at the door has no predecessor or successor, i.e., $r.pred = r.succ = (\bot, \bot)$ and $r.avoid = 0$.
We redefine the eligible port.

\begin{definition}{\textbf{(Eligible Port from $v$)}}
\label{def:eligible-general-single}
A port $p_v$ incident with $v$ is called an eligible port for $r$ positioned at $v$ if the following conditions hold.
(i) $p_v \neq p'_1$, where $r.succ = (p'_1, p'_2)$.
(ii) There exists at least one unoccupied vertex $2$ hops away from $v$ along the port $p_v$.
(iii) There is no robot $r'$ with $r'.color \neq $ \texttt{FINISH} present at most $3$ hops away from $v$ along $p_v$. 
\end{definition}

\begin{figure}[h]
    \centering
    \includegraphics[width=0.8\linewidth]{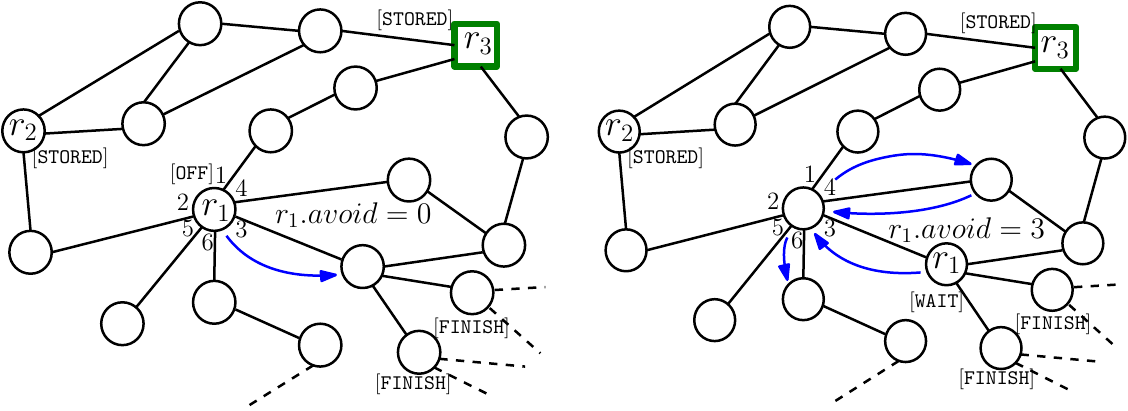}
    \caption{$r_1$ is the head and $r_1.succ = r_2$ and $r_2.succ = r_3$. (left) $r_1$ moves along $p_v^{min} = 3$ (right) $r_1$ backtracks multiple times due to presence of $r_3$ at $3$ hop neighbour}
     \label{fig:backtrack}
\end{figure}

\noindent We distinguish the following cases based on the position and current color of $r$.

\noindent $\blacktriangleright$~ \textbf{Case 1 ($r.color $ = \texttt{OFF} and $r$ is situated on the door):}
If $\exists p_v occupied$ $(v_{nbr}^1, p_v, \neq\texttt{FINISH})$, $r$ does nothing.
Elseif $\exists p_v \exists v_{nbr}^2 occupied(v_{nbr}^2, p_v, \neq\texttt{FINISH})$ $r$ stores $r.pred = (p_1, p_2)$, where $p_1 =p_v$ and $p_2$ is the port of $e(v_{nbr}^1, v_{nbr}^2)$. $r$ then changes its color to \texttt{STORED} and remains in place until the robot on $v_{nbr}^2$ moves further or changes its color to \texttt{FINISH}.
Otherwise, if $\exists p_v \exists v_{nbr}^2~ \neg occupied(v_{nbr}^2, p_v,$ $ \sim)$, $r$ computes the minimum eligible port $p_v^{min}$. 
It now sets its target $v_{nbr}^1$ along the port $p_v^{min}$ and stores $r.succ = (p'_1, \bot)$, where $p'_1$ is the port of the edge $e(v_{nbr}^1, v)$. Then, it moves to $v_{nbr}^1$ with the current color.
If such $p_v^{min}$ does not exist but $\exists p_v \neg occupied(v^1_{nbr}, p_v, \sim)$, $r$ changes its color to \texttt{FINISH}.

\noindent $\blacktriangleright$~ \textbf{Case 2 ($r.color =$ \texttt{OFF} or \texttt{WAIT} and $r$ is not on the door):}
    In this case, $r$ may or may not be the head.   
    So, we further divide this case into the following subcases.
    In case of $r$ being the head, and situated more than $1$ hop away from the door (Case 2.2), it might switch its color to \texttt{WAIT}.
    
    \begin{itemize}[left=0pt]
        \item \textbf{Case 2.1 ($r.pred = (\bot, \bot) \text{ and } r.succ=(p_1', \bot)$):} It means that $r$ is the head and $1$ hop away from the door along $p'_1$. 
        We slightly modify Definition \ref{def:eligible-general-single} for such a robot, where the computation of the minimum eligible port $p_v^{min}$ does not disregard a port $p_v$ even when a $3$ hop neighbour of $v$ along $p_v$ is occupied by (which is the door).
        If $p_v^{min}$ exists, $r$ updates $r.succ$ as $(p''_2, p'_1)$ and moves to the vertex $v_{nbr}^1$ along $p_v^{min}$ with the current color, where $p''_2$ is the port of the edge $e(v_{nbr}^1, v)$.
        Otherwise, $r$ sets its color \texttt{FINISH} without any movement.

        \item \textbf{Case 2.2 ($r.pred = (\bot, \bot) \text{ and } r.succ = (p_1', p_2')$):} 
        This indicates that the head $r$ is situated at a vertex at least $2$ hops away from the door. Let $w$ be the neighbour of $v$ along $p_1'$. 
        Based on the values of the local variables, we segregate the following sub-cases.

        \begin{itemize}
            \item \textbf{$r.color =$ \texttt{OFF}:} $r$ waits until $occupied(w_{nbr}^1, p_2', \texttt{STORED})$ is satisfied. 
            When the predicate becomes true, $r$ computes $p_v^{min}$. 
            If $p_v^{min}$ exists, it moves to $v_{nbr}^1$ along $p_v^{min}$ after updating $r.succ = (p''_2, p'_1)$ and $r.color =$ \texttt{WAIT}, where $p''_2$ is the port of the edge $e(v_{nbr}^1, v)$ (refer to Fig. \ref{fig:backtrack} (up)). If instead $p_v^{min}$ does not exist, it terminates after changing its color to \texttt{FINISH}.

            \item \textbf{$r.color =$ \texttt{WAIT} and $r.avoid = 0$:} In this case, $r$ might have to backtrack in the direction of its successor.
            $r$ calculates $p_v^{min}$ and if it exists, $r$ turns its color to \texttt{OFF} without any movement. 
            If instead, $p_v^{min}$ does not exist, and the predicate $\exists p  (\neg occupied(v_{nbr}^1, p , \sim) \land \forall v_{nbr}^2 occupied(v_{nbr}^2, p, $ $\texttt{FINISH}))$ holds true, $r$ switches its color to \texttt{FINISH} without any movement. 
            Otherwise, if $\exists ~occupied(v_{nbr}^3, p_v, col\}$ is in $col \in \{\texttt{OFF}, \texttt{STORED}\}$, $r$ moves to $v_{nbr}^1$ with the current colour and after updating $r.avoid = p^*$, where $v_{nbr}^1$ is the neighbour of $v$ along the port $p'_1$ and $p^*$ is the port of the edge $e(v_{nbr}^1, v)$ (see Fig. \ref{fig:backtrack} (down)).
            We call the movement toward $r.succ$ as backtracking.

            \item \textbf{$r.color =$ \texttt{WAIT} and $r.avoid = p^* > 0$:} $r$ first computes all the eligible ports $p > p^*$ and chooses the minimum among them as $p_{v}^{min}$.
            If no such port exists, it changes its color to \texttt{FINISH} without any movement. Else, it moves to $v_{nbr}^1$ along $p_v^{min}$ after updating $r.succ = (p''_2, p'_2)$, $r.avoid = 0$ and $r.color =$ \texttt{WAIT}, where $p''_2$ is the port of $e(v, v_{nbr}^1)$ incident to $v_{nbr}^1$.
        \end{itemize}
        
        \item \textbf{Case 2.3 ($r.pred = (p_1, \bot)$):}
        In this case, $r$ has partial information about the position of its predecessor.
        So, it needs to find the exact position of the predecessor.
        Let $u = v_{nbr}^1$, a $1$ hop neighbour of $v$ along the port $p_1$. 
        A robot $r'$ must exist with $r'.color \neq \texttt{FINISH}$  at a vertex $2$ hops away from $v$ along the port $p_1$.
        Let $p_2$ be the port of the edge $e(u, r')$, incident to $u$.
        $r$ updates $r.pred = (p_1, p_2)$ and changes its color to \texttt{STORED} without any movement.
        We prove in our analysis that there can not be two such robots $r'$ along $p_1$.   
    \end{itemize}

    \noindent $\blacktriangleright$~ \textbf{Case 3} ($r.color =$ \texttt{STORED}):
    Since the current color of $r$ is \texttt{STORED}, there exists $p_1, p_2$ such that $r.pred = (p_1, p_2)$.
    The value of $r.succ$ depends on its distance from the door.
    If $r$ is on the door, $r.succ = (\bot,\bot)$. 
    If $r$ is $1$ hop away, $r.succ = (p'_1, \bot)$ and if at least $2$ hops away, $r.succ = (p'_1, p'_2)$ for some ports $p'_1$ and $p'_2$.
    Let $w = v_{nbr}^1$ be the neighbour of $v$ along $p'_1$. 
    If $\exists w_{nbr}^1~ occupied(w_{nbr}^1, p'_2, \texttt{STORED})$ does not hold, $r$ remains stationary.
    Otherwise, regardless of the position of $r$, it considers the vertex $u$, a neighbour of $v$ along the port $p_1$.
    If the predicate $occupied(u_{nbr}^1, p_2, \neq \texttt{FINISH}) \lor occupied(u_{nbr}^2, p_2,  \texttt{WAIT})$ holds, $r$ maintains the status quo.
    If $occupied(u_{nbr}^1, p_2, \texttt{FINISH})$ holds, it updates $r.pred = (\bot, \bot)$ and $r.color=\texttt{OFF}$, thus becoming the new head.
    However, if $\neg occupied(u_{nbr}^1, p_2, \sim)$ holds but $occupied(u_{nbr}^2, p_2, \texttt{WAIT})$ does not hold, $r$ sets $u$ as the target.  
    Finally, $r$ moves to $u$ after updating $r.color = \texttt{OFF}, r.pred = (p_2, \bot)$ and $r.succ =(p,\bot) \text{ or } (p, p_1')$, depending on whether it is on the door or not, where $p$ is the port of the edge $e(u, v)$.
    The termination condition is the same as earlier.

\noindent \textbf{Execution Example of the Algorithm \textsc{\textmd{Graph\_SingleDoor}}}
This section illustrates the execution of the algorithm \textsc{Graph\_SingleDoor}, which requires the robots to have $3$ hops visibility, $4$ colors and $O(\log \Delta)$ memory. 
The algorithm achieves MVC on a general graph with a single door. 
Fig. \ref{fig:ex-general-1} to \ref{fig:ex-general-12} depicts the process of the algorithm \textsc{Graph\_SingleDoor}. 
We keep the captions of the following figures self-explanatory.

\begin{figure}

\centering
\begin{minipage}[b]{0.48\linewidth}
\centering
     \includegraphics[width=0.9\linewidth]{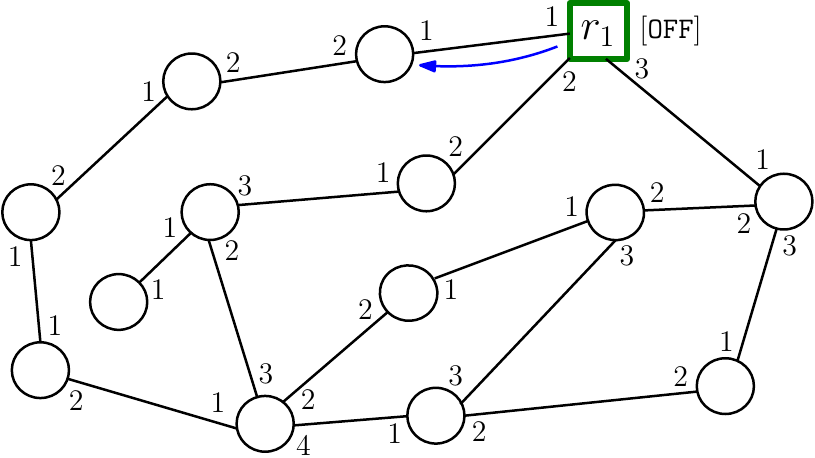}
     \caption{Initially, the robot $r_1$ is on the door with the color \texttt{OFF} and targets a neighbouring vertex along the minimum eligible port $1$.}
     \label{fig:ex-general-1}
\end{minipage}
\hfill
\begin{minipage}[b]{0.48\linewidth}
\centering
         \includegraphics[width=0.9\linewidth]{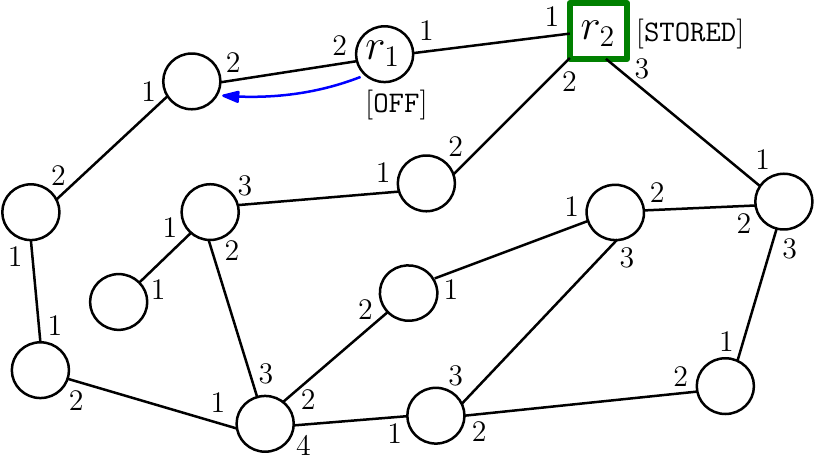}
         \caption{The head $r_1$ finds $r_2$ as its successor with \texttt{STORED} and moves ahead similarly by calculating minimum eligible port.}
         \label{fig:ex-general-2}
\end{minipage}
\begin{minipage}[b]{0.48\linewidth}
\centering
     \includegraphics[width=0.9\linewidth]{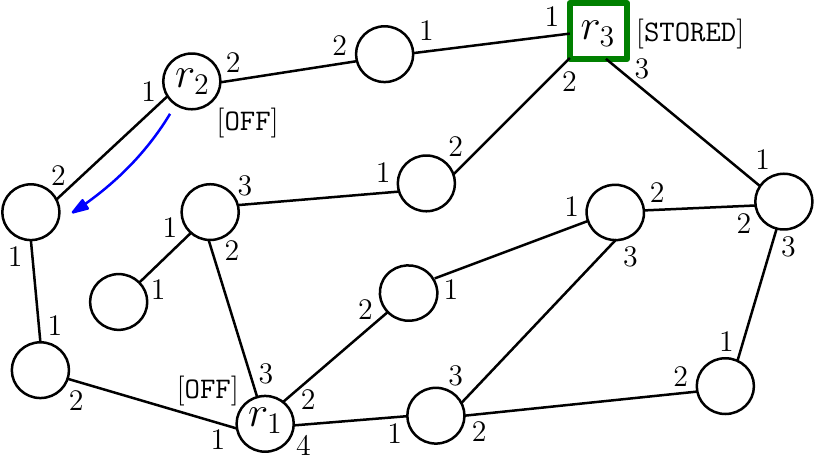}
     \caption{$r_2$ finds its successor $r_3$ with \texttt{STORED} and moves in the direction of its predecessor $r_1$ while $r_1$ waits till $r_2$ arrives at a $2$ hop neighbour.}
     \label{fig:ex-general-3}
\end{minipage}
\hfill
\begin{minipage}[b]{0.48\linewidth}
\centering
         \includegraphics[width=0.9\linewidth]{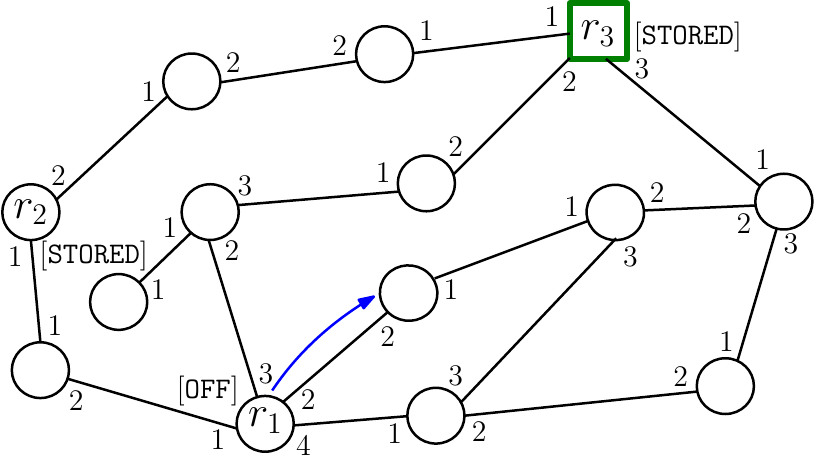}
         \caption{$r_2$ sitting at a $2$ hop neighbour of $r_1$ stores its ports corresponding to $r_1$ in $r_2.pred$ while $r_1$ selects a vertex as target along port $2$.}
         \label{fig:ex-general-4}
\end{minipage}

\begin{minipage}[b]{0.48\linewidth}
\centering
     \vspace{1.5em}
     \includegraphics[width=0.9\linewidth]{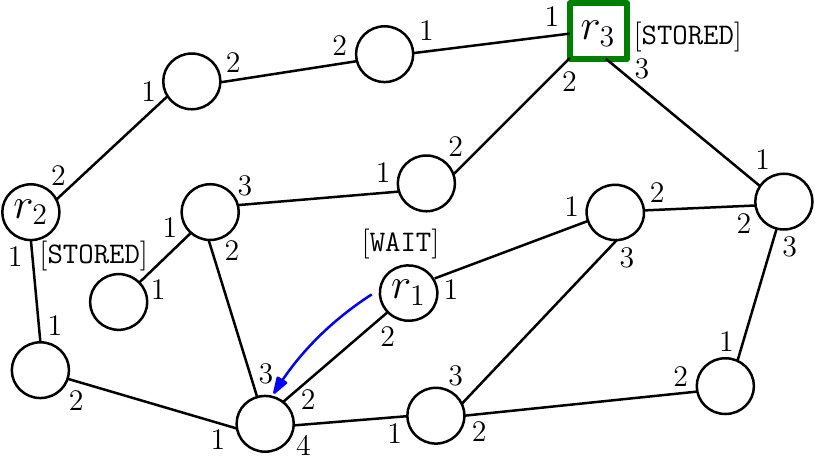}
     \caption{$r_1$ finds $r_3$ at a $3$ hop neighbour with color $\neq \texttt{FINISH}$ and backtracks to its previous position and stores the port $2$ in $r.avoid$.}
     \label{fig:ex-general-5}
\end{minipage}
\hfill
\begin{minipage}[b]{0.48\linewidth}
\centering
         \includegraphics[width=0.9\linewidth]{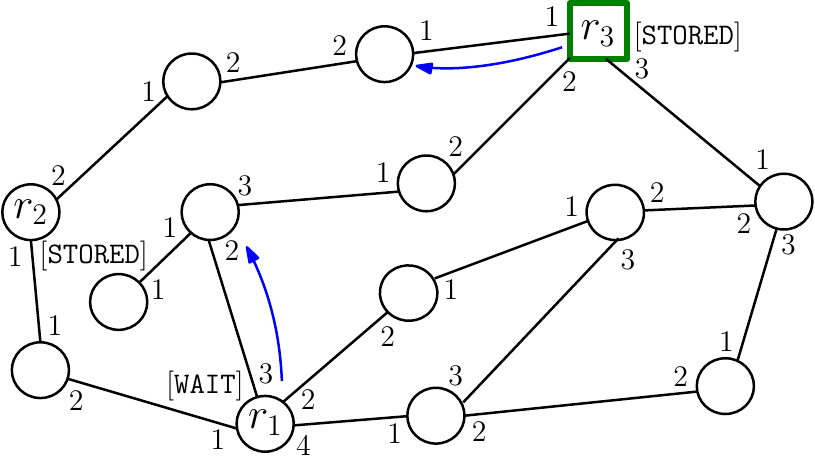}
         \caption{$r_1$ discards all the ports $\leq r.avoid$ and moves to the neighbour along the port $3$ while $r_2$ maintains status quo seeing $r_1$ with \texttt{WAIT}.}
         \label{fig:ex-general-6}
\end{minipage}

\end{figure}

\begin{figure}

\begin{minipage}[b]{0.48\linewidth}
\centering
     \includegraphics[width=0.9\linewidth]{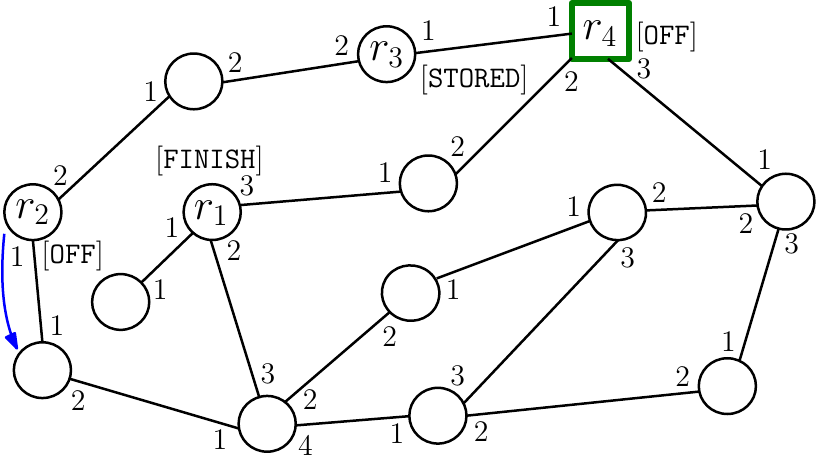}
     \caption{$r_1$ finds no eligible port but sees an unoccupied neighbour along the port $1$ with no other $2$ hop neighbour in that direction. So, it switches to \texttt{FINISH} and its successor $r_2$ becomes the new head.}
     \label{fig:ex-general-7}
\end{minipage}
\hfill
\begin{minipage}[b]{0.48\linewidth}
\centering
         \includegraphics[width=0.9\linewidth]{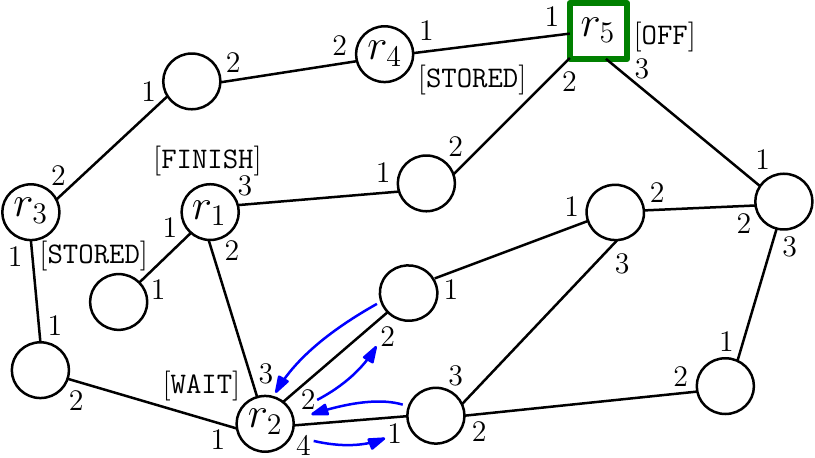}
    \caption{$r_2$ first moves along port $2$ with the color \texttt{WAIT} and returns back after seeing $r_5$ with the color $\neq \texttt{FINISH}$. Similarly, it also moves along $4$ and comes back with $r.avoid = 4$, while $r_3$ remains stationary. }
         \label{fig:ex-general-8}
\end{minipage}

\begin{minipage}[b]{0.48\linewidth}
\centering
     
     \includegraphics[width=0.9\linewidth]{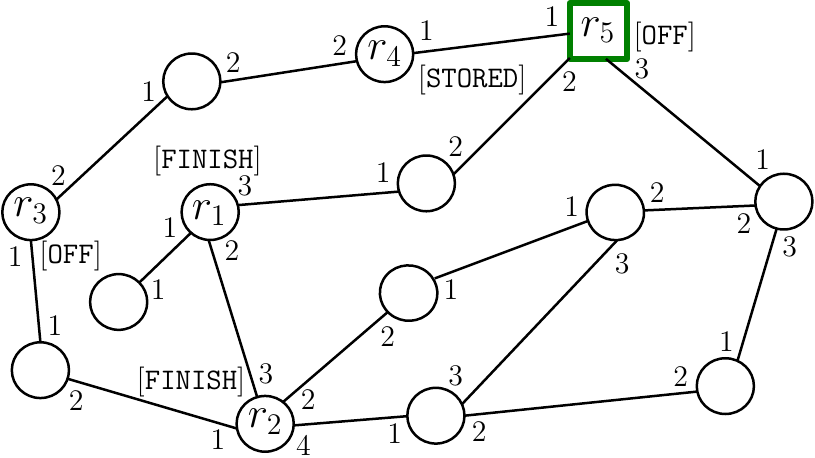}
     \caption{$r_2$ turns its color to \texttt{FINISH} as no eligible port $> r.avoid = 4$ exists. So, $r_3$ becomes the new head by changing its color to \texttt{OFF} from \texttt{STORED}.}
     \label{fig:ex-general-9}
\end{minipage}
\hfill
\begin{minipage}[b]{0.48\linewidth}
\centering
         \includegraphics[width=0.9\linewidth]{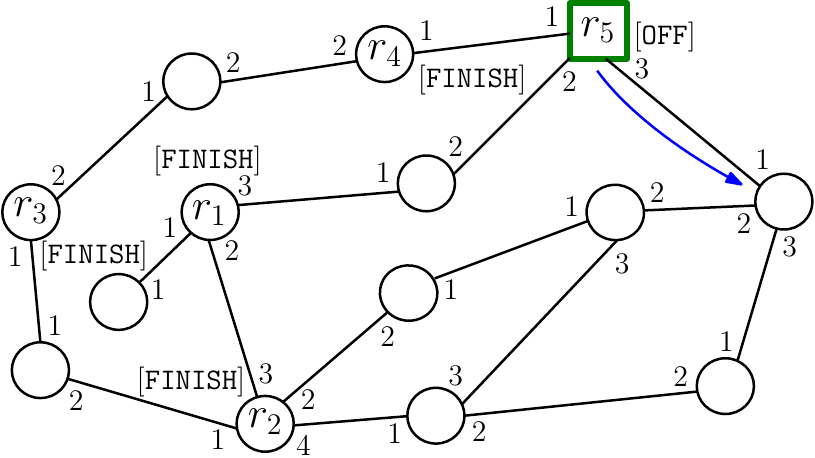}
         \caption{Both $r_3$ and $r_4$ change their color to \texttt{FINISH} after finding no eligible port from their respective positions. $r_5$ at the door becomes the new head and moves along the port $3$.}
         \label{fig:ex-general-10}
\end{minipage}

\begin{minipage}[b]{0.48\linewidth}
\centering
     \vspace{1.5em}
     \includegraphics[width=0.9\linewidth]{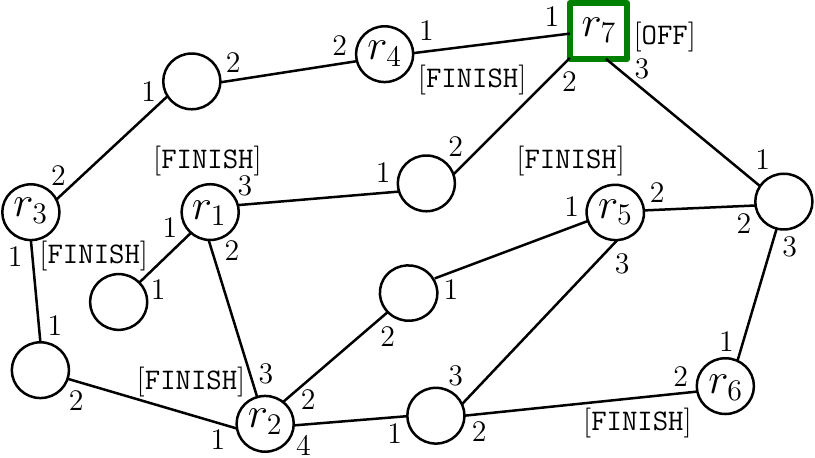}
     \caption{Similarly as earlier, $r_5$ and $r_6$ settle at proper vertices of the graph. Afterwards, $r_7$ becomes the new head sitting at the door.}
     \label{fig:ex-general-11}
\end{minipage}
\hfill
\begin{minipage}[b]{0.48\linewidth}
\centering
         \includegraphics[width=0.9\linewidth]{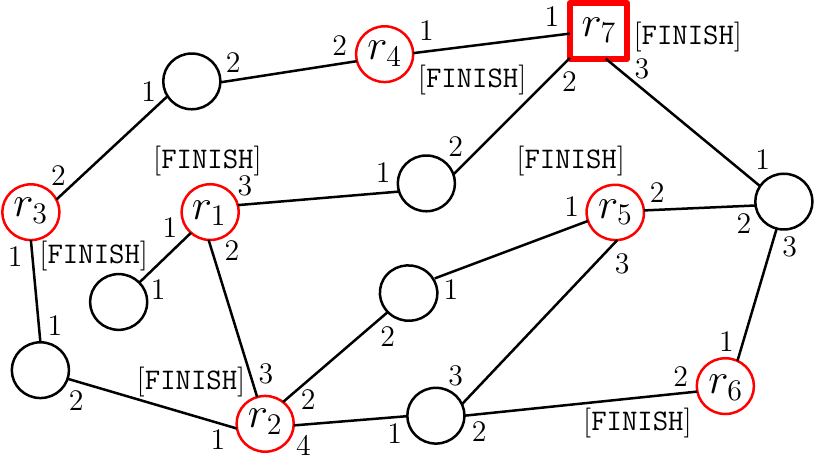}
         \caption{$r_7$ finds no eligible port to move to but has an unoccupied neighbour of the door. So, it switches to \texttt{FINISH} and the algorithm terminates.}
         \label{fig:ex-general-12}
\end{minipage}
\end{figure}

\noindent \textbf{Deferred Analysis of \textsc{\textmd{Graph\_SingleDoor}} Algorithm:}
Here, we analyse the correctness and time complexity of the algorithm \textsc{Graph\_SingleDoor} with detailed proofs. Several arguments closely parallel those presented in Section \ref{subsec:tree-singledoor}.
We first show that each robot maintains a predecessor and successor relationship before termination (Lemma \ref{lem:unique_pred} and \ref{lem:unique-succ}). 
Once established, we can restate the following results, as previously mentioned in Section \ref{subsec:tree-singledoor}.

\begin{itemize}
    \item Robots with color $\neq$\texttt{FINISH} collectively form a chain-like structure during the exploration of the graph. (Similar arguments as in Lemma \ref{lem:chain_like_structure}).
    
    \item The head of the chain terminates before other members of the chain. 
    Upon its termination, its immediate successor becomes the new head of the chain (Remark \ref{remark:chain_of_robots_tree}) and the rest of the robots continue to be the members of the current chain. 
    
    \item Moreover, any two consecutive robots (except the robot on the door) in the chain maintain a distance of at most $3$ hops and at least $2$ hops from each other (Similar as Lemma \ref{lemma:max_min_distance-btw-robots}).
\end{itemize}

\begin{figure}[h]
        \centering
        \includegraphics[width=0.5\linewidth]{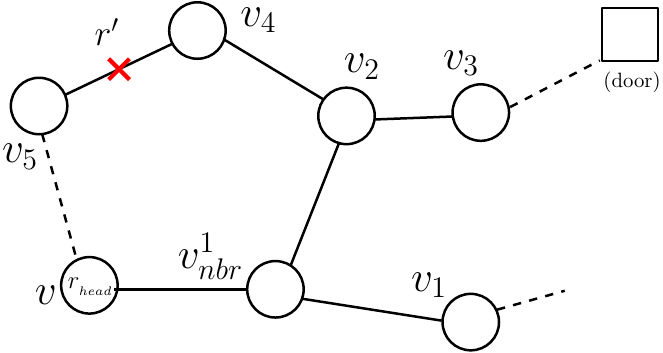}
        \caption{When $r'$ is on the edge $e(v_4, v_5)$, the follower of $r'$, say $r''$, must be at $v_3$, So, $r_h$ can see $r''$ and will not move to $v_{nbr}$}
        \label{fig:correctly_find_pred}
    \end{figure}

\begin{lemma}
    \label{lem:unique_pred}
    For any robot $r$ positioned at vertex $v$, if $r.pred = (p_1, \bot)$, then there exists exactly one robot, with color \texttt{OFF} or \texttt{STORED} located at $2$ hops away from $v$ along the port $p_1$, i.e.,  $r$ can correctly determine the port $p_2$ such that $r.pred$ lies exactly at $2$ hops away from $v$ along the port $p_1$ and then $p_2$.
\end{lemma}

\begin{proof}
    We prove this by contradiction. 
    Suppose that at some time $t$, $r$ finds two robots, $r_1$ at $v_1$ and $r_2$ at $v_2$, with color \texttt{OFF} or \texttt{STORED}, such that both are $2$ hops away from $v$ along the port $p_1$. 
    Let $v^1_{nbr}$ be the $1$ hop neighbour of $v$ along the port $p_1$. 
    Without loss of generality, assume that the correct predecessor of $r$ is $r_1$, and the robot $r$ is the first robot that fails to correctly identify its predecessor at the time $t$.
    That is, until time $t$, all robots correctly maintain the predecessor-successor relationship, forming a chain originating from the door.
    Consider the time $t' (< t)$ when the head $r_{head}$ is at the vertex $v$ and decides to move to $v^1_{nbr}$. 
    If $v_2$ were a door vertex, it would necessarily be occupied by a robot, leading to a contradiction. 
    Therefore, $v_2$ cannot be a door vertex, which implies that $r_2$ must have both a successor and a predecessor.
    Let $v_3$ and $v_4$ be the neighbours of $v_2$, where $v_4$ lies along $r_2$'s successor direction and $v_3$ along its predecessor (as shown in Fig.~\ref{fig:correctly_find_pred}). 
    By the definition of an eligible port, robot $r_{head}$ cannot have any visible neighbours within $3$ hops other than its successor. 
    Therefore, at time $t'$, there cannot be any robot positioned on $v_3$, $v_4$, or on the edges $e(v_2, v_3)$ and $e(v_2, v_4)$.
    Since $t'<t$, all the robots with color $\neq \texttt{FINISH}$ form a chain-like structure at time $t'$.
    Hence, by Lemma \ref{lemma:max_min_distance-btw-robots}, the distance between two successive robots on the chain is at most 3, implying that a robot $r'$ must lie on the edge $e(v_3, v_5)$ or on $v_5$.
    When $r'$ is at the vertex $v_4$ at time $t'' (< t')$ and selects $v_5$ as its target, it must have observed its predecessor at $v_3$. 
    This contradicts the configuration at time $t'$, where no robot occupies $v_3$. \qed
\end{proof}

\begin{lemma}
    \label{lem:unique-succ}
    A robot $r$ can correctly determine the successor that lies $2$ hops away from its current location.
\end{lemma}

\begin{proof}
    A robot positioned on the door does not have successor and sets $r.succ = (\bot, \bot)$. When such a robot $r$ targets the vertex $v_1$, it updates $r.succ = (p, \bot)$ before moving. After reaching $v_1$, it can correctly determine its successor (the robot on the door), which is located one hop away along the port $p$.
    Therefore, if the current position of $r$, say $v$, is one hop away from the door, $r$ can accurately identify its successor.
    Next, we consider the case where $v$ is at least two hops away from the door.
    Let $v^{prev}$ be the previous position of $r$ and at that time $r.succ = (p'_1, p'_2)$. Without loss of generality, we further assume that $r$ always identifies its successor correctly before reaching $v$.
    Thus, when $r$ was at $v^{prev}$, its successor was two hops away from $v^{prev}$ along the port $p'_1$ and then $p'_2$.
    Before moving to $v$, the robot $r$, lying on $v^{prev}$, updates $r.succ$ by $r.succ = (p, p'_1)$, where $p$ is the port of the edge $e(v_t, v)$ incident to $v_t$. Since the successor of $r$ correctly calculates its predecessor (which is $r$ itself) and follows it (by Lemma \ref{lem:unique_pred}), $r$ will eventually arrive two hops away from $v$, allowing it to determine the correct position of the successor, which is two hops away from $v$ along the port $p$ and then $p'_1$. \qed
\end{proof}

\begin{corollary}
    \label{cor:forming-chain-graph}
     Starting from the door, all the robots, which are not with the color \texttt{FINISH}, lie on a path (which we refer to as a chain).  
\end{corollary}

\begin{proof}
    The statement is based directly on the above Lemma \ref{lem:unique_pred}, which guarantees that the predecessor-successor relationship is correctly established for all robots currently with color $\neq$ \texttt{FINISH}. \qed
\end{proof}

\begin{corollary}
    \label{cor:chain-prpoerty}
    Any path of three consecutive vertices of a current chain (path consisting of all the robots with color $\neq$\texttt{FINISH}) contains at least one robot.
\end{corollary}

\begin{proof}
    On the contrary, let's assume that $P=\{v_1, v_2, v_3\}$ be the path on the chain that does not possess any robot with color $\neq$\texttt{FINISH}. Let $v_3$ be the farthest vertex from the door among all the vertices of $P$ in the current chain. Let $r$ be the robot that is currently located on $v_4$ or on edge $e(v_3, v_4)$, where $v_4$ is the next vertex after $v_3$ in the current chain. Consider the time when $r$ was at $v_3$ and decide to move to $v_4$. At this time, the successor of $r$ must be at $v_1$, which is a contradiction. \qed
\end{proof}

\begin{lemma}
    \label{lem:head_never_visit-very_oldVertex}
    The head robot of the current chain never revisits a vertex that it has previously visited in an earlier LCM cycle, excluding its immediate previous position.
\end{lemma}

\begin{proof}
    We prove the lemma by contradiction.
    Let $r$ be the head robot of the current chain.
    In contrast, let us assume that the head robot $r$, currently located at the vertex $v$, selects the vertex $v_{nbr}^1$ as its next target, which is not its immediate previous position, but a vertex already visited by itself in one of its previous LCM cycles.
    This movement is either part of the backtrack procedure or is just a forward movement.
    If it is the forward movement and is previously visited by $r$, then $v_{nbr}^1$ must be on the current chain. 
    Since the door vertex is always occupied, $v_{nbr}^1$ cannot be the door.
    Hence, there must be two neighbours of $v_{nbr}^1$, say $v_2$ and $v_3$, which is part of the current chain.
    By Corollary \ref{cor:chain-prpoerty}, at least one of $v_2$ or $v_3$ must be occupied by a robot with color $\neq \texttt{FINISH}$. Consequently, the port of the edge $e(v, v_{nbr}^1)$ incident to $v$ becomes ineligible for $v$, and $r$ cannot move to $v_{nbr}^1$, which is a contradiction.
    Therefore, we remain to prove when $r$ decides to move to $v_{nbr}^1$ from $v$ in the backtracking process, i.e., $r$ has visited $v_{nbr}^1$ from $v$ previously.
    In this process, when a robot first moves to $v'$, a neighbour of $v$, and then returns to $v$ again, it updates $r.avoid = p'$, where $p'$ is the port of the edge $e(v, v')$ incident to $v$.
    After this, when $r$ again decides to move to another neighbour of $v$, say $v''$, it only considers those ports that have a value greater than $r.avoid = p'$. 
    Consequently, the second target $v''$ in this backtracking procedure is not the first target $v'$. 
    When $r$ again returns to $v$ from $v''$, it updates $r.avoid = p'' > p'$, where $p''$ is the port of the edge $e(v, v'')$ incident to $v$.
    Hence, if the current target of $r$ is $v_{nbr}^1$ along the port $p$, it could not be the vertex $v'$ or $v''$ or any other vertex, which was the previous target of $r$ from $v$, as $p > r.avoid \geq p'' > p'$, which is again a contradiction.    \qed
\end{proof}

\begin{corollary}
    \label{cor:never-visit-old-vertex-graph}
    A robot $r$, which is not the head of the current chain, never revisits a vertex that it has already visited in one of its previous LCM cycles.
\end{corollary}

\begin{proof}
    By Lemma \ref{lem:unique-succ}, since $r$ correctly identifies its successor $r.succ = (p'_1, p'_2)$, it excludes the port $p'_1$ from the set of eligible ports and thus never returns to the vertex $v^{prev}$ from which it arrived at $v$.
   Moreover, as all robots in the current chain follow the head, and the head never revisits any previously occupied vertex except its immediate predecessor (by Lemma \ref{lem:head_never_visit-very_oldVertex}), no robot in the chain revisits such vertices. \qed
\end{proof}

\begin{lemma}
\label{lem:min_distance_between_two_consec_robots}
    The distance between two consecutive robots on the current chain (excluding the robot on the door) is at least two. This distance is measured along the ports from where it reaches the current position.
\end{lemma}

\begin{proof}
    The claim holds trivially for any two robots neither of which is the head.
    This is because a robot never moves back toward its successor and always verifies the presence of its predecessor before moving. Specifically, a robot remains stationary when its predecessor is two hops away along the stored path in $r.\text{pred}$.
    Now consider the case where $r_1$ is the head robot and $r_2$ is its successor. 
    If $r_1.color=\texttt{OFF}$, then by Lemma~\ref{lem:unique-succ}, $r_1$ correctly identifies its successor $r_2$ and never moves back in its direction.
    If $r_1.color = \texttt{WAIT}$, then $r_1$ may backtrack. However, such a backtracking move occurs from a $3$ hop neighbour from the current position of $r_2$ to a $2$ hop neighbour.
    In the meantime, $r_2$, upon seeing $r_1$ as of color \texttt{WAIT}, remains in place.
    As a result, the distance between $r_1$ and $r_2$ remains at least two. \qed
\end{proof}

\begin{theorem}
     \label{thm:asyncGenSS}
     Algorithm \textsc{Graph\_SingleDoor} fills an MVC of an arbitrary graph $G$ with \texttt{FINISH}-colored robots having $3$ hops visibility, $4$ colors and $O(\log \Delta)$ memory in $O(|E|)$ epochs under $\mathcal{ASYNC}$ and without collision. 
\end{theorem}

\begin{proof}
    
    By Lemmas \ref{lem:unique_pred}, \ref{lem:unique-succ}, \ref{lem:min_distance_between_two_consec_robots} and Corollary \ref{cor:forming-chain-graph}, \ref{cor:chain-prpoerty}, all the active robots (robots with the color $\neq \texttt{FINISH}$) form a chain-like structure, where the distance between the two consecutive robots is at most three and at least two. 
    These are analogous to those established in the Lemma \ref{lem:chain_like_structure},\ref{lemma:max_min_distance-btw-robots},\ref{lem:not_visit_old_edges} and Remark \ref{remark:chain_of_robots_tree}.
    All the results together along with Lemma \ref{lem:head_never_visit-very_oldVertex} and Corollary \ref{cor:never-visit-old-vertex-graph}, the algorithm remains free from collision, as noted in Remark \ref{rem:collision-free}. 

    During the execution of the algorithm, a robot uses $4$ colors: \texttt{OFF, STORED, WAIT} and \texttt{FINISH}. 
    Moreover, every robot $r$ stores the port numbers $r.pred, r.succ$ and $r.avoid$ to remember the location of the predecessor, successor and a port to avoid visiting in future. This requires $O(\log \Delta)$ memory in total, since the maximum value of a port incident to a vertex is the maximum degree $\Delta$ of the graph $G$.
    
    \textit{Time Complexity:} 
    Using a similar argument as in the proof of Lemma \ref{lem:time_asyncSS}, we have the identical result on the time complexity. 
    The key distinction lies in how robots identify their predecessor and successors. In \textsc{Tree\_SingleDoor} algorithm, a robot with color \texttt{OFF-i} determines them by observing the robots with the colors \texttt{OFF-(i+1)} and \texttt{OFF-(i-1)}, whereas, in \textsc{Graph\_SingleDoor} algorithm, $r$ determines those two robots by computing $r.pred$ and $r.succ$, which are correctly computed as shown in \ref{lem:unique_pred}, \ref{lem:unique-succ}.
    Additionally, each robot $r$ waits one extra epoch for its successor $r.succ$ to change its color to \texttt{STORED}, adding at most $|E|$ epochs to the total time.
    The only thing remaining in analysing the time complexity is when a robot $r$ backtracks during the execution of the algorithm.
    In this process, it is possible that the head robot $r$ might revisit a particular vertex multiple times. 
    For example, let us assume that the current vertex $v$ of $r$ is of degree $\Delta$ and $v_i$ is a neighbour of $v$ for $1 \leq i \leq \Delta$ along the port $i$ (which is incident to $v$).
    In the worst case, $r$ first moves to $v_1$ with color \texttt{WAIT} and then backtracks to $v$ with \texttt{WAIT} after finding a $3$ hop neighbour of $v_1$ occupied by an active robot. 
    It then moves to $v_2$ and similarly observes an active robot on a $3$ hop neighbour of $v_2$, hence backtracks to $v$ again. 
    Similarly, this back-track process might continue for $\Delta' (\leq \Delta-1)$ neighbours of $v$, except for the neighbour (say $v_{\Delta}$) in the direction of its successor. 
    Without loss of generality, we assume that the set of all neighbours from where $r$ backtracks to $v$ is $\{v_1, v_2, \cdots, v_{\Delta'}\}$.
    Since every visit of $r$ to a neighbour $v_i (1 \leq i \leq \Delta')$ and backtracking to $v$ together requires at most $2$ epochs, the process of backtracking to $v$ from its neighbours requires at most $2\Delta'$ epochs in total.
    If in case $\Delta'= \Delta -1$, the robot $r$ switches its color to \texttt{FINISH}.
    Otherwise, $r$ moves to one of the neighbours of $v$, where it changes its color either to \texttt{OFF} or \texttt{FINISH}.
    In either case, $r$ needs at most two more epochs. 
    After this, the successor of the head $r$ either becomes eligible to move toward $r$ or becomes the new head, depending on whether the current color of $r$ is \texttt{OFF} or \texttt{FINISH}.
    Therefore, the whole chain might wait for at most $(2 \Delta' + 2) \leq 2 \Delta$ epochs due to the head's backtracking.
    At first glance, this may appear to contradict the claim in the proof of Lemma \ref{lem:time_asyncSS}, which states that head takes $O(k)$ epochs to explore $k$ many edges.
    Recall that in the algorithm \textsc{Tree\_SingleDoor}, an edge is considered explored if the head visits at least one endpoint of the edge for the first time.
    For the algorithm \textsc{Graph\_SingleDoor}, we slightly refine the definition of the \textit{explored edge}.
    An edge is said to be explored if the head visits at least one of its endpoints for the first time and does not subsequently backtrack from it.
    The proof of the claim in Lemma \ref{lem:time_asyncSS} was based on the worst-case scenario where each movement of the head results in exploring only one edge.
    However, in the case that we described above, when $r$ reaches the vertex $v$, all its $\Delta'$ neighbours $\{v_i : 1\leq i \leq \Delta'\}$ are unoccupied. Consequently, all the edges $\{e(v, v_i): 1 \leq i \leq \Delta'\}$ are explored by the head $r$ for the first time.
    Therefore, even if $r$ takes $(2\Delta'+2)$ epochs to reach a termination at the vertex  $v$ or move to a neighbouring vertex from $v$ and sets the color \texttt{OFF}, it finds $\Delta'$ many explored edges in between the two movements.
    As a consequence, the average waiting time of the head $r$ per edge explored by it is $\frac{2\Delta'+2}{\Delta'} \leq 4$.
    Thus we still have the claim: ``head takes $O(k)$ epochs to explore $k$ edges".
    Hence, the overall time complexity of \textsc{Graph\_SingleDoor} remains $O(|E|)$ asymptotically.

    We now prove the feasibility and minimality of the solution returned by the Algorithm \textsc{Graph\_SingleDoor}.

    \textit{Feasibilty:} We prove by contradiction. Let $e=(u,v)$ be an edge that is not covered by any \texttt{FINISH}-colored robot. Such existence robot the chain neither arrives at $u$ nor $v$.
    The edge $e$ can not be the edge incident to the door, as the termination condition ensures that such an edge is always covered.
    Without loss of generality, let us assume that $e$ is the only uncovered edge. Then, at least one of the neighbours of $u$ must be occupied with a \texttt{FINISH}-colored robot. 
    Let $r$ be the last robot to terminate (with the color \texttt{FINISH}) among all the robots occupying the neighbours of $u$.
    Consider the time $t$, when $r$, lying on one of the neighbours of $u$, decides to terminate. By our assumption, both $u$ and $v$ must be unoccupied at $t$.
    $r$ decides to terminate despite seeing its 2-hop neighbour $v$ unoccupied. This implies that either $u$ has some one/two hop neighbours or $v$ has some neighbours occupied by $\neq \texttt{FINISH}$-colored robots. 
    In either case, we treat such a robot as the new $r$ and again inspect the time $t'>t$ when this robot terminates.
    By continuing this argument recursively, we eventually reach a situation where a robot $r$ finds both of $u$ and $v$ unoccupied, and no other \texttt{FINISH}-colored robots exist along the port in the direction of edge $e(u,v)$. Despite this, $r$ marks that port as ineligible and terminates, which contradicts its own termination condition. 

    \textit{Minimality:} We prove the minimality of the solution by contradiction. 
    Suppose, for the sake of contradiction, that a robot $r$ at the vertex $v$ violates the minimality. 
    This implies that all the neighbours of $v$ are occupied by robots with the color \texttt{FINISH}.
    It is evident that the vertex $v$ cannot be the door vertex, as a robot terminates at the door when it finds at least one unoccupied neighbour of the door.
    Therefore, $r$ must have completed at least one movement before reaching $v$. Let $v_1$ be the previous position of $r$, which is also the neighbour of $v$.
    When $r$ is the head of the chain and decides to terminate, its successor must be at most two hops away from the current position $v$ along the direction of the vertex $v_1$. 
    By Remark \ref{remark:chain_of_robots_tree} the successor robot of $r$, say $r_1$, must terminate after the robot $r$.
    We first deal with the case where $v_1$ is the door.
    In this case, $r.succ=(p, \bot)$, where $p$ is the port of the edge $e(v, v_1)$ incident to $v$.
    When $r$ is at the door $(v_1)$ and decides to reach $v$, it must observe an unoccupied neighbour of $v$, say $v_2$.
    Since $r$ decides to terminate after reaching $v$, all neighbours of $v_2$ (excluding $v$) are occupied by \texttt{FINISH}-colored robots.
    Hence, once $r$ turns to \texttt{FINISH}, all the neighbours of $v_2$ are becoming occupied with \texttt{FINISH}-colored robots, and no robot can subsequently reach $v_2$, a contradiction.

    We now discuss the case where $v_1$, the previous position of $r$ before it reaches $v$, is not the door vertex.
    There can be two sub-cases depending on the cause of movement of $r$ from $v_1$. 
    (Case-I) $r$ being the head robot moves to $v$ from $v_1$ along its eligible port and terminates. (Case-II) $r$ being the head robot backtracks towards its predecessor and then terminates.
    We first establish the contradiction for the Case-I.
    Consider the time $t_1$, when $r$, situated at $v_1$, decides to move to $v$.
    If all the neighbours of $v$ are occupied at $t_1$, the port $p_1$ of the edge $e(v_1, v)$ incident to $v_1$ would be ineligible for the robot $r$.
    As a consequence, $r$ would not reach $v$.
    Therefore, there exists at least one unoccupied neighbour of $v$, say $v_2$, at the time $t_1$. 
    Since $r$ terminates without backtracking, all ports of $v$ are ineligible, meaning all neighbours of $v_2$ are occupied by \texttt{FINISH}-colored robots at termination.
    But then no robot can ever reach $v_2$, a contradiction.
    Now we present the contradiction for Case-II, where $r$ terminates after the backtracks from $v_1$.
    It means that all other port $p'$ incident to $v_1$, except $p_1$, are ineligible either because there is an active robot located two or three hops away from $v_1$ along $p'$, or because the one-hop neighbour of $v_1$ along $p'$ are occupied by \texttt{FINISH}-colored robot.
    Consequently, if a robot arrives at $v_1$ after the time $t_1$, it must have come from a neighbour along which active robots were present either two or three hops away from $v_1$ at the time $t_1$. 
    These positions are symmetric to vertex $v$ from the perspective of $v_1$, any robot arriving at $v_1$ from these positions would simply backtrack to its previous location.
    Thus, after $t_1$, no robot ever settles at $v_1$, contradicting the assumption that $v$ violates the minimality.
    Hence, the proof. \qed
\end{proof}

\subsection{Algorithm (\textsc{\textmd{Graph\_MultiDoor}}) for Multiple Doors}

\label{subsec:graph-multidoor}

In this algorithm, the definitions of $r.pred$, $r.succ$, $r.avoid$ and \textit{head} remain same as in Section \ref{subsec:graph-singledoor}.
The algorithm uses 4 hop visibility and $4H+1$ colors: \{\texttt{OFF}$_h$, \texttt{move?}$_h$, \texttt{STORED}$_h$, \texttt{WAIT}$_h$ and \texttt{FINISH}\}, for $1 \leq h \leq H$.
Initially, each robot $r$ on the $h$-th lowest ID door ($door_h$) is with color \texttt{OFF}$_h$ and $r.pred = r.succ = (\bot, \bot), r.avoid = 0$, $1\leq h \leq H$.
Here, we modify the definition of eligible port.

\begin{definition}{\textbf{(Eligible Port from $v$)}}
A port $p_v$ is called an eligible port for $r$, positioned at $v$ with $r.color \in  \{\texttt{OFF}_h, \texttt{move?}_h, \texttt{STORED}_h, \texttt{WAIT}_h\}$, if the following conditions hold.
(i) $p_v \neq p'_1$, where $r.succ = (p'_1, p'_2)$.
(ii) There exists at least one unoccupied vertex $v_{nbr}^2$ along $p_v$.
(iii) No 3-hop neighbour of $v$ along port $p_v$ is occupied by a robot $r'$ with $r'.\texttt{color} \in \{\texttt{OFF}_h, \texttt{STORED}_h, \texttt{WAIT}_h\}$.
(iv) $r$ does not find two robots $r'$ at $v'$ and $r''$ at $v''$ such that $v_{nbr}^1$ or $v_{nbr}^2$ along the port $p_v$ lies on the path between $v'$ and $v''$, where $r'.color = \texttt{col}_{h'}$ and $r''.color =$ \texttt{col}$'_{h'}$, for some \texttt{col}, \texttt{col}$'$ $\in \{\texttt{OFF}, \texttt{move?}, \texttt{STORED} \}$, and $h' \neq h$. 
    
    
\end{definition}

\noindent We distinguish the cases based on the position and current color of $r$.
Several actions align with those in \textsc{Graph\_SingleDoor}, with minor adjustments to color transitions, specifically, each color \texttt{col} in Section~\ref{subsec:graph-singledoor} corresponds here to $\texttt{col}_h$, where $h$ denotes the ID of the door through which the robot entered.

\noindent $\blacktriangleright$~ \textbf{Case 1} ($r.color=$ \texttt{OFF}$_h$ and $r$ is on $door_h$):
If $\exists p_v ~ occupied(v_{nbr}^1, p_v, \neq \texttt{FINISH})$ $\lor$ $\exists p_v \exists v_{nbr}^2~ occupied(v_{nbr}^2, p_v, \texttt{OFF}_h)$ holds, $r$ follows the action as mentioned in Case 1, Section \ref{subsec:graph-singledoor}.
Otherwise if $\exists p_v \exists v_{nbr}^2~ \neg occupied(v_{nbr}^2, p_v, \sim)$, $r$ finds $p_{v}^{min}$ and proceeds with the following steps sequentially.
Let $u=v_{nbr}^1$ along $p_{v}^{min}$. 
Let $p_u$ be the smallest port incident to $u$ such that the 1-hop neighbour along it is unoccupied.
Let $p'$ be the port of the edge $e(u, v)$.

    \begin{itemize}
        \item[C1] \textsc{If} $\exists u_{nbr}^3~ occupied(u_{nbr}^3, p_u, \texttt{col}_{h'}^{i'})$ with $h > h'$ holds, $r$ remains in place.

        \item[C2] \textsc{ElseIf} $\exists u_{nbr}^1  occupied(u_{nbr}^1, p_u, \texttt{col}_{h'})$ for $h>h'$ and \texttt{col} $\in \{$\texttt{OFF}, \texttt{move?}, \texttt{STORED}, \texttt{WAIT} $\}$, $r$ does nothing.

        \item[C3] \textsc{ElseIf} $\neg exists(p_u, 2)~ \lor ~\forall u_{nbr}^2 ~ (\neg occupied(u_{nbr}^2, p_u, \sim) \lor occupied(u_{nbr}^2, p_u,$ $ \texttt{FINISH}))$, $r$ moves to $u$ with the current color after updating $r.succ = (p', \bot)$.

        \item[C4] \textsc{ElseIf} $\exists u_{nbr}^2 \neg occupied(u_{nbr}^2, p_u, \sim) \land \exists u_{nbr}^2  occupied $  $(u_{nbr}^2, p_u, \texttt{col}_{h'})$ with $\texttt{col}$ $ \in \{\texttt{OFF}, \texttt{move?}, \texttt{WAIT}\}$ and $h$ is smallest among all such $h'$, $r$ stays put.

        \item[C5] \textsc{ElseIf} $\exists u_{nbr}^2~ \neg occupied(u_{nbr}^2, p_u, \sim) \land \exists u_{nbr}^2 ~ occupied$ $(u_{nbr}^2, p_u, \texttt{OFF}_{h'})$ such that $h$ is greater than the minimum among all such $h'$, $r$ switches to \texttt{move?}$_h$.

        \item[C6] \textsc{ElseIf} $\forall v_{nbr}^2 (\neg occupied(v_{nbr}^2, p_v^{min}, \sim) \lor (occupied$ $(v_{nbr}^2, $ $ p_v^{min}, \texttt{OFF}_{h'}) \land (h < h')))$, $r$ moves to $u$ with its current color after updating $r.succ = (p', \bot)$.

        \item[C7] \textsc{ElseIf} $p_v^{min}$ does not exist and $\exists p_v \neg occupied(v^1_{nbr}, p_v, \sim)$, $r$ turns to \texttt{FINISH}.

    \end{itemize}

    \noindent $\blacktriangleright$~ \textbf{Case 2} ($r.color =$ \texttt{OFF}$_h$ or $\texttt{WAIT}_h$ and $r$ is not on any door): We divide this case into the following sub-cases, based on whether $r$ is a head robot or not. 

    \begin{itemize}[left=0pt]
        \item \textbf{Case 2.1} ($r.pred = (\bot, \bot)$): 
        When $r.succ = (p'_1, \bot)$ for some port $p'_1$, i.e., $r$ is one hop away from the door, it proceeds as described in Case 2.1, Section \ref{subsec:graph-singledoor}.
        When $r.succ = (p'_1, p'_2)$, some adjustments are made to prevent collisions with robots from the other door.
        If $r.color = \texttt{WAIT}_h$ and $r.avoid = 0$, it follows Case 2.2, Section \ref{subsec:graph-singledoor}.
        If $r.color = \texttt{OFF}_h$, $r$ considers the vertex $w=v_{nbr}^1$ along $p'_1$ and waits until $occupied(w_{nbr}^1, p_2', \texttt{STORED})$ is satisfied. Once true, it computes $p_v^{min}$. 
        If $r.color = \texttt{WAIT}_h$ and $r.avoid = p^* > 0$, $r$ considers all the ports $p>p^*$ and selects the minimum eligible port $p_v^{min}$ among them.
        In the last two cases, if no $p_v^{min}$ exists, $r$ switches to \texttt{FINISH}; otherwise, $r$ identifies $u, p_u$ and $p'$ as defined in Case 1. If $\exists u_{nbr}^1 occupied(u_{nbr}^1, p_u, \texttt{move?}_{h'})$, $r$ waits until $u^1_{nbr}$ changes to \texttt{OFF}$_{h'}$. \textsc{ElseIf}, $r$ follows the steps C1, C2, C4 and C5 sequentially. 
        \textsc{Elseif} $\forall v_{nbr}^2 (\neg occupied(v_{nbr}^2, p_v^{min}, \sim) \lor (occupied $ $(v_{nbr}^2, p_v^{min}, \texttt{OFF}_{h'}) \land (h < h')))$, $r$ moves to $u$ with \texttt{WAIT}$_h$ after updating $r.avoid = 0$ and $r.succ = (p', p'_1)$ or $(p', p'_2)$ depending on the current color of $r$ is $\texttt{OFF}_h$ or $\texttt{WAIT}_h$.
        
        \end{itemize}

        

        




    \begin{itemize}[left=0pt]
        \item \textbf{Case 2.2} ($r.pred = (p_1, \bot)$):
        Let $u = v_{nbr}^1$ along the port $p_1$. 
        $r$ first finds $r'$, a 2-hops neighbour of $v$ along $p_1$ and with $r'.color \in \{\texttt{OFF}_h, \texttt{STORED}_h, \texttt{move?}_h,$ $ \texttt{WAIT}_h\}$.
        Let $p_2$ be the port of the edge $e(u, r')$, incident to $u$.
        $r$ updates $r.pred = (p_1, p_2)$ and changes its color to \texttt{STORED}$_h$ without any movement.
    \end{itemize}
   
    \noindent $\blacktriangleright$~ \textbf{Case 3} ($r.color=$ \texttt{STORED}$_h$):
    The action mirrors Case 3 of Subsection \ref{subsec:graph-singledoor}, where $r$ either moves toward its predecessor or becomes the new head.

    \noindent $\blacktriangleright$~ \textbf{Case 4} ($r.color =$ \texttt{move?}$_h$):
    Let $u = v_{nbr}^1$ along the port $p_v^{min}$. 
    If all neighbours of $u$, except $v$, are unoccupied or occupied by robots with colors in $\{\texttt{OFF}_{h'}, \texttt{move?}_{h'}, \texttt{STORED}_{h'}, \texttt{WAIT}_{h'}, \texttt{FINISH}$\} for $h' > h$, it moves to $u$ with the color \texttt{WAIT}$_h$. 
    However, if $r$ finds $u_{nbr}^1$ as occupied with \texttt{OFF}$_{h'}$-colored robot with $h' < h$, it switches to \texttt{OFF}$_h$ without movement. 
    Otherwise, $r$ does nothing.

The termination condition remains the same as earlier.

\vspace{2em}

\noindent \textbf{Analysis of \textsc{\textmd{Graph\_MultiDoor}} Algorithm:} Our detailed analysis establishes the following theorem.

\begin{theorem}%
\label{theorem:Graph_MultiDoor}
    Algorithm \textsc{Graph\_MultiDoor} fills MVC of an arbitrary graph $G$ having $H$ doors with \texttt{FINISH}-colored robots having $4$ hops visibility, $O(H)$ colors and $O(\log \Delta)$ memory in $O(|E|)$ epochs under $\mathcal{ASYNC}$ and without collision. 
\end{theorem}
\begin{proof}
    We prove the theorem in four parts: collision-freeness, time and memory complexity, feasibility of the solution, and the minimality of the solution returned by the Algorithm \textsc{Graph\_MultiDoor}.

    \textit{Collision-free:}
    We have a total of $H$ chains, one for each door. 
    For each of the chains starting from a door, we can restate the following lemma as established in Section \ref{subsec:tree-singledoor}.
    Lemma \ref{lem:head_never_visit-very_oldVertex} confirms that the head of a chain, that starts from a door, never revisits any previously visited vertex, except its immediate previous position.
    Corollary \ref{cor:never-visit-old-vertex-graph} ensures that robots on a chain other than the head do not revisit the previously visited vertex. 
    Corollary \ref{cor:chain-prpoerty} and Lemma \ref{lem:min_distance_between_two_consec_robots} guarantee that the distance between two consecutive robots on a chain is at least two and at most three.
    From all these results, it is evident that the two robots from the same chain do not collide with each other.    
    Additionally, two robots from the different chain do not collide from the definition of eligible port and Theorem \ref{thm:coll_tree_MS}.
    
    \textit{Time and Memory Complexity:}
    In this analysis, we retain the definition of the eligible edge as given in the earlier Theorem \ref{theorem:tree-singledoor}, where we an edge is considered explored if one of the endpoints is visited by the head of a chain for the first time and does not subsequently backtrack from it.
    A head of a chain explores edges incrementally until it encounters another chain.
    Lemma \ref{lem:time_asyncSS} and Theorem \ref{theorem:tree-singledoor} show that until the time of such an encounter, the average number of epochs required per explored edge remains constant.
    When two chains meet with each other, we have a different observation as presented in Lemma \ref{lemma:time-complexity-tree-multidoor}.
    Let $r_h$ and $r_{h'}$ be two robots, originating from the $h$-th and $h'$-th lowest ID door, respectively.
    Without loss of generality, we assume that $r_h$ has encountered $r_{h'}$, a robot different from its own chain, along the minimum eligible port.
    Then, by a similar argument as established in Lemma \ref{lemma:time-complexity-tree-multidoor}, either $r_h$ or $r_{h'}$ proceed further for the exploration (or covering) process in at most two epochs. 
    This concludes that, at any time, there must be a chain which explores the graph at a constant average epoch per edge exploration, even if the other chain remains halt at that time.
    Consequently, in the worst case, we need to explore $|E|$ many edges and thus $O(|E|)$ epochs to reach the termination.
    The algorithm uses $(4H+1) = O(H)$ colors listed as: $\{\texttt{OFF}_h, \texttt{STORED}_h, \texttt{move?}_h, \texttt{WAIT}_h, \texttt{FINISH}\}$.
    Furthermore, each robot $r$ stores at most $5$ port values incident to its current position or one of its neighbours: two for $r.pred$, two for $r.succ$, and one for $r.avoid$. Thus, each robot $r$ requires $5 \log \Delta = O(\log \Delta)$ bits of memory, in addition to its fixed color set.
    
    \textit{Feasibility:}
    We prove the feasibility of the solution by contradiction.
    On the contrary, we assume that there exists an uncovered edge $e = (u, v)$, i.e., both endpoints of $e$ are unoccupied at termination.
    By our assumptions, both $u$ and $v$ can not be the door vertices.
    We first consider the case where one of the endpoints of $e$ (say $u$) is adjacent to a door, while the other endpoint $(v)$ is not.
    \begin{itemize}
        \item If $v$ has an unoccupied neighbour other than $u$, then a robot must leave the door vertex $u$ and move to $v$.
        \item If all neighbours of $v$ are occupied with \texttt{FINISH}-colored robots, then the robot on $u$ must eventually turn its color to \texttt{FINISH} after finding that all the ports incident to $u$ are ineligible, but there exists an unoccupied neighbour $v$.
    \end{itemize}
    
    In either case, a contradiction arises.
    Now we consider the case where neither $u$ nor $v$ is a door vertex. 
    Without loss of generality, suppose that $u$ is closer to one of the doors than $v$.
    If all the neighbours of $u$ are unoccupied, then we consider a neighbour $w'$ of $u$ that is closer to a door than $u$.
    If $w'$ is a door vertex, then the previous case applies, which leads to a contradiction.
    So $w'$ must not be a door vertex.
    In this scenario, we treat $e=e(w', u)$ as the uncovered edge instead of $e(u, v)$ and check whether one of the neighbours of $w'$ is occupied by a \texttt{FINISH}-colored robot.
    If not, we again replace $(w', u)$ by an edge incident to $w'$ that is closer to one of the doors than $w$.
    We repeat this process until we find that one of the endpoints of $e$ is a door or one of the neighbours of the endpoints of $e$ is occupied by a \texttt{FINISH} robot.
    If we find that one of the endpoints of $e$ is a door, then we are done.
    Therefore, we assume that one of the neighbours of $u$, say $u_{nbr}^1$, is occupied by a \texttt{FINISH}-colored robot $r$.
    Consider the time $t$, when $r$ lying at $u_{nbr}^1$ decides to terminate.
    At $t$, both $u$ and $v$ must be unoccupied.
    Despite seeing an unoccupied 2-hop neighbour $v$, $r$ terminates at $u_{nbr}^1$, which implies that one of the following conditions must hold:
    \begin{itemize}
        \item (i) $r$ finds an active robot $r'$ at $v'$ that is a $3$ hop neighbour of $u_{nbr}^1$ in the direction of $u$, and both $r$ and $r'$ are from the same door.
        \item (ii) $r$ finds two robots $r'$ at $v'$ and $r''$ at $v''$ such that $u$ or $v$ lies on the path between $v'$ and $v''$, and both $r'$ and $r''$ are from the same door $(door_{h'}$ but different door than $r$ $(door_h)$.
    \end{itemize}
    When the first case holds, using an argument similar to that in the \textit{feasibility} proof of Theorem~\ref{theorem:tree-singledoor}, we can show that the head of the chain originating from $door_h$ eventually reaches $v'$.
    From there, it becomes eligible to move to $u$ or $v$, thus covering the edge $e = (u, v)$.
    In the latter case, assume that $r''$ is the predecessor of $r'$. 
    Then, the head of the chain from $door_{h'}$ will eventually reach $v'$ and settle at $u$, covering the edge $e = (u, v)$. 
    Hence, in any of the cases mentioned above, the head of a chain must reach and settle either to $u$ or to $v$, which is a contradiction.

    \textit{Minimality:} We prove this property by contradiction.
    Let us assume that the robot $r$, lying on $v$, violates minimality.
    Then all the neighbours of $v$ must be occupied by \texttt{FINISH}-colored robots.
    Using a similar argument as presented in the \textit{minimality} proof of Theorem \ref{theorem:tree-singledoor}, $v$ can not be a door vertex.
    Therefore, $r$ must complete at least one movement and it can not be the last robot to set its color to \texttt{FINISH} after all its neighbours turn their color to \texttt{FINISH}.
    Hence, after $r$ terminates, a robot $r'$ must eventually reach at some neighbour $v_1$ of the vertex $v$.
    Let $p_1$ be the port of the edge $e(v, v_1)$ incident to $v$,
    and $v_2$ be the previous position of $r'$, before the robot $r'$ reaches $v_1$.
    When $r'$ is at $v_2$, and decides to move to $v_1$, it must see a neighbour of $v_2$, say $v_3$, as unoccupied.
    Then $v_3$ is a two hop neighbour of $v$ along the port $p_1$.
    If $r'$ is from the same door as of $r$, then $r'$ would never reach $v_1$, as argued in Theorem \ref{theorem:tree-singledoor}.
    Therefore, we assume that $r$ and $r'$ are from the different doors, $door_h$ and $door_{h'}$, respectively, where $h \neq h'$.
    Consider the time $t$, when $r$ terminates at $v$.
    At $t$, $r$ must find all its incident port as ineligible, including $p_1$.
    This implies that one of the following must hold:    \begin{itemize}
        \item (i) The vertex $v_1$ is part of the chain that originates from a door other than $door_{h}$.
        \item (ii) One of the neighbours of $v_1$ (assume that it is $v_3$) belongs to a different chain than the one to which $r$ belongs.
        \item (iii) there exists an active robot $r''$ at 3-hop neighbour along the port $p_1$ (i.e., at a 2-hop neighbour of $v_1$), such that both $r$ and $r''$ are from the same door $door_h$.
    \end{itemize}
    For the case (i) $r$ would never have moved to $v$ from its previous position upon detecting that $v_1$ is part of the different chain than its own.
    For the case (ii), the robot $r'$ would not move from $v_2$ to $v_1$ because $v_3$, one hop neighbour of $v_2$, lies on a different chain.
    For the case (iii) $r'$ would never target $v_1$ or reach $v_2$, as it finds $v_1$ as the part of the chain that originates from $door_h$ and $r'$ itself from the different door $(door_{h'})$.
    Hence, in either case, we have the contradiction.
    \qed
\end{proof}

\section{Concluding Remarks}

We are particularly interested in investigating a fundamental question in a barebone model: Can a group of mobile robots achieve a particular graph property with only local or limited knowledge?  
We affirmatively answered this question by considering the \emph{filling MVC problem} on an arbitrary graph with robots operated under the $\mathcal{ASYNC}$ scheduler and with 4 hops visibility.
The main highlight of this paper is to be able to achieve a deployment strategy that is time-optimal even under $\mathcal{ASYNC}$ scheduler. 
Interestingly, our strategy achieves MinVC for trees with the optimal time and memory. 
%
We believe a similar strategy can give time optimality in filling MIS  problem \cite{PRAMANICK2023114187}.
Our method improves efficiency by allowing alternate robots on the chain to move in constant epochs, thereby significantly reducing the overall time complexity. 
As a future direction, one can also investigate other graph properties in the same model.

\bibliographystyle{splncs04}
\bibliography{ALGOWIN}

\begin{thebibliography}{10}
\providecommand{\url}[1]{\texttt{#1}}
\providecommand{\urlprefix}{URL }
\providecommand{\doi}[1]{https://doi.org/#1}

\bibitem{10.1007/978-3-540-92862-1_11}
Barrameda, E.M., Das, S., Santoro, N.: Deployment of asynchronous robotic
  sensors in unknown orthogonal environments. In: Fekete, S.P. (ed.)
  Algorithmic Aspects of Wireless Sensor Networks ALGOSENSORS '08. pp.
  125--140. Springer Berlin Heidelberg, Berlin, Heidelberg (2008)

\bibitem{10.1007/978-3-642-45346-5_17}
Barrameda, E.M., Das, S., Santoro, N.: Uniform dispersal of asynchronous
  finite-state mobile robots in presence of holes. In: Flocchini, P., Gao, J.,
  Kranakis, E., Meyer auf~der Heide, F. (eds.) Algorithms for Sensor Systems,
  ALGOSENSORS '14. pp. 228--243. Springer Berlin Heidelberg, Berlin (2014)

\bibitem{barriere2011uniform}
Barriere, L., Flocchini, P., Mesa-Barrameda, E., Santoro, N.: Uniform
  scattering of autonomous mobile robots in a grid. International Journal of
  Foundations of Computer Science  \textbf{22}(03),  679--697 (2011)

\bibitem{10.1007/978-3-031-48882-5_10}
Chand, P.K., Molla, A.R., Sivasubramaniam, S.: Run for cover: Dominating set
  via mobile agents. In: Georgiou, K., Kranakis, E. (eds.) Algorithmics of
  Wireless Networks, ALGOWIN '23. pp. 133--150. Springer Nature Switzerland,
  Cham (2023)

\bibitem{ELOR2011783}
Elor, Y., Bruckstein, A.M.: Uniform multi-agent deployment on a ring.
  Theoretical Computer Science  \textbf{412}(8),  783--795 (2011)

\bibitem{10.1007/3-540-46632-0_10}
Flocchini, P., Prencipe, G., Santoro, N., Widmayer, P.: Hard tasks for weak
  robots: The role of common knowledge in pattern formation by autonomous
  mobile robots. In: Algorithms and Computation. pp. 93--102. Springer Berlin
  Heidelberg, Berlin, Heidelberg (1999)

\bibitem{1200625}
Heo, N., Varshney, P.: A distributed self spreading algorithm for mobile
  wireless sensor networks. In: 2003 IEEE Wireless Communications and
  Networking, 2003. WCNC 2003. vol.~3, pp. 1597--1602 vol.3 (2003)

\bibitem{hideg2017uniform}
Hideg, A., Lukovszki, T.: Uniform dispersal of robots with minimum visibility
  range. In: International Symposium on Algorithms and Experiments for Sensor
  Systems, Wireless Networks and Distributed Robotics, ALGOSENSORS '17. pp.
  155--167. Springer (2017)

\bibitem{hideg2020asynchronous}
Hideg, A., Lukovszki, T.: Asynchronous filling by myopic luminous robots. In:
  Algorithms for Sensor Systems: 16th International Symposium on Algorithms and
  Experiments for Wireless Sensor Networks, ALGOSENSORS '20, Pisa, Italy,
  September 9–10, 2020, Revised Selected Papers. p. 108–123.
  Springer-Verlag, Berlin, Heidelberg (2020)

\bibitem{10.1023/A:1019625207705}
Howard, A., Matari\'{c}, M.J., Sukhatme, G.S.: An incremental self-deployment
  algorithm for mobile sensor networks. Auton. Robots  \textbf{13}(2),
  113–126 (Sep 2002)

\bibitem{Hsiang2004}
Hsiang, T.R., Arkin, E.M., Bender, M.A., Fekete, S.P., Mitchell, J.S.B.:
  Algorithms for Rapidly Dispersing Robot Swarms in Unknown Environments, pp.
  77--93. Springer Berlin Heidelberg, Berlin, Heidelberg (2004).
  \doi{10.1007/978-3-540-45058-0\_6}

\bibitem{kamei2020asynchronous}
Kamei, S., Tixeuil, S.: {An Asynchronous Maximum Independent Set Algorithm By
  Myopic Luminous Robots On Grids}. The Computer Journal  (11 2022).
  \doi{10.1093/comjnl/bxac158}, bxac158

\bibitem{Karp1972}
Karp, R.M.: Reducibility among Combinatorial Problems, pp. 85--103. Springer
  US, Boston, MA (1972). \doi{10.1007/978-1-4684-2001-2\_9}

\bibitem{10.1145/3631461.3631543}
Pattanayak, D., Bhagat, S., Gan~Chaudhuri, S., Molla, A.R.: Maximal independent
  set via mobile agents. In: Proceedings of the 25th International Conference
  on Distributed Computing and Networking. p. 74–83. \textit{ICDCN '24},
  Association for Computing Machinery, New York, NY, USA (2024).
  \doi{10.1145/3631461.3631543}

\bibitem{1307146}
Poduri, S., Sukhatme, G.: Constrained coverage for mobile sensor networks. In:
  IEEE International Conference on Robotics and Automation, 2004. Proceedings.
  ICRA '04. 2004. vol.~1, pp. 165--171 Vol.1 (2004)

\bibitem{PRAMANICK2023114187}
Pramanick, S., Samala, S.V., Pattanayak, D., Mandal, P.S.: Distributed
  algorithms for filling mis vertices of an arbitrary graph by myopic luminous
  robots. Theoretical Computer Science  \textbf{978},  114187 (2023).
  \doi{https://doi.org/10.1016/j.tcs.2023.114187}

\end{thebibliography}

\end{document}